\documentclass{article}
\usepackage[utf8]{inputenc}
\usepackage{authblk}
\usepackage{setspace}
\usepackage{placeins}
\usepackage{amsthm}
\usepackage{amssymb}  
\usepackage{pifont}
\newtheorem{remark}{Remark}
\usepackage{comment}
\usepackage{bm}
\usepackage{amsthm}
\usepackage{amsmath}
\theoremstyle{plain}
\newtheorem{definition}{Definition}
\theoremstyle{plain} 
\newtheorem{theorem}{Theorem}
\newtheorem{lemma}{Lemma}
\usepackage{longtable}
\usepackage{lipsum}
\usepackage[export]{adjustbox}
\usepackage{array,multirow,graphicx}
\usepackage{float}
\usepackage{tikz}
\usetikzlibrary{trees}
\usepackage[margin=0.9 in]{geometry}
\usepackage{graphicx}
\graphicspath{ {./figures/} }
\usepackage{changepage}
\usepackage{amsmath}
\usepackage{amsfonts}
\usepackage{longtable}
\usepackage{listings}
\usepackage{graphicx}
\usepackage{longtable}
\usepackage{tabularx,booktabs}
\usepackage{lipsum}
\DeclareUnicodeCharacter{2212}{-}
\DeclareUnicodeCharacter{03B7}{$\eta$}
\usepackage{ragged2e}
\usepackage{graphicx}
\usepackage{algorithm}
\usepackage{booktabs}
\usepackage{textcomp}
\usepackage{algpseudocode}
\usepackage{bm}
\usepackage{xcolor}
\usepackage{subfigure}
\usepackage{float}
\usepackage[utf8]{inputenc}
\usepackage[nottoc]{tocbibind} 
\usepackage{amssymb}
\usepackage{amsmath}
\newtheorem{problem}{Problem}

\title{Influence Maximization on  Temporal Networks with Persistent and Reactive Behaviors}

\author{Aaqib Zahoor \footnote{Corresponding Author}}
\author{Iqra Altaf Gillani}
\author{Janib ul Bashir}

\affil{Department of Information Technology, NIT Srinagar, India}
\affil{\{aaqib\_phaite003, iqraaltaf, janibbashir\}@nitsri.ac.in}

\date{}

\begin{document}

\maketitle

\begin{abstract}
Influence Maximization (IM) aims to identify a set of $k$ seed nodes in a network that maximizes information spread, with applications in viral marketing, epidemic containment, and behavioral targeting. While most prior work assumes static networks, real-world networks are inherently temporal and shaped by user behaviors such as temporary disengagement, reactivation, and reinforcement from repeated exposures. To address these dynamics, we introduce \textit{cpSI-R}, a novel diffusion model that jointly captures temporal evolution and behavioral persistence. The model allows nodes to temporarily deactivate and later reactivate under peer influence, while also reinforcing adoption likelihood through repeated interactions. We prove that the resulting influence spread function is monotone and submodular, making it amenable to efficient greedy approximation. Building on this model, we develop an efficient influence maximization framework featuring a structure-aware temporal sampling technique.

We validate cpSI-R through empirical comparisons with classical diffusion models, demonstrating its consistent ability to achieve greater influence spread across temporal snapshots. Comprehensive experiments on various datasets show that our strategy significantly outperforms state-of-the-art baselines in both influence spread and runtime, with particularly strong gains on large-scale networks. These results highlight the practical utility of cpSI-R and our scalable IM framework in modeling and leveraging realistic diffusion dynamics over evolving social systems.
\end{abstract}

\section{Introduction}

Social networks serve as critical platforms for the rapid dissemination of information, profoundly shaping public discourse, marketing strategies, and political engagement. With the growing reliance on online platforms for communication, understanding and predicting influence propagation has become a focal point of research. A prominent application is word-of-mouth marketing, where businesses strategically select seed users to initiate a diffusion process that maximizes product adoption. By leveraging social influence, companies aim to amplify engagement and optimize outreach, ensuring that a small, carefully chosen set of individuals can drive large-scale adoption of ideas or products.

The influence maximization (IM) problem, formalized by Kempe et al. \cite{kempe}, seeks to identify the most influential subset of users in a network to maximize the expected spread of influence. Since its introduction, IM has been extensively studied, leading to advancements in adaptive seeding strategies, influence estimation techniques, and scalable optimization methods \cite{viralmarketing}\cite{adaptiveim}\cite{scalableim}\cite{tbcelf}\cite{Influentialnodetracking}. However, while classical IM frameworks excel in static networks, real-world social networks are inherently dynamic—characterized by evolving connections, interactions, and trends. These temporal dynamics significantly impact seed selection strategies, necessitating models that account for the evolving nature of influence spread. Although recent studies have addressed this limitation in online settings \cite{online1}\cite{online2}\cite{online3}\cite{online4}, the challenge of optimal seed selection in temporal networks persists.

A key limitation of classical diffusion models is their assumption that once a node becomes inactive or fails to activate, it cannot be influenced again. These models treat influence spread as a one-shot, irreversible process, overlooking the reality that individuals often cycle through phases of engagement, disengagement, and reactivation. For example, in political campaigns, a user may initially engage by sharing content, later lose interest, and then re-engage following a major debate or controversy. Similarly, in public health efforts, individuals who ignore early vaccine awareness messages may later change their stance due to peer influence or updated scientific information. These scenarios highlight that influence propagation is not a static or linear event but a dynamic, evolving process shaped by repeated exposure and reactivation

The inclusion of \emph{reactivation} and \emph{reinforcement dynamics} introduces additional complexity in influence estimation. As nodes oscillate between active and inactive states, the expected spread function exhibits variations that must be meticulously modeled. This behavior complicates optimal seed selection and increases computational costs \cite{azaouzi2021new} \cite{aral2018social}. Furthermore, influence propagation in temporal networks often deviates from classical IM assumptions, making optimization tasks more challenging. Effective seed selection requires models that account for persistence, reactivation, and reinforcement effects to capture influence propagation with higher fidelity.

To address these challenges, we introduce the \emph{Continuous Persistent Susceptible-Infected Model with Reinforcement and Re-activation} (cpSI-R), which explicitly models influence persistence, node reactivation, and reinforcement-driven spread. Unlike conventional models that assume one-time activation, cpSI-R captures the cumulative effect of multiple exposures over time, allowing nodes to regain influence potential upon reactivation. The model extends classical approaches by incorporating three key behavioral dynamics: influence persistence, where influence from an active node is sustained over time; reactivation, where previously inactive nodes may become active again through peer influence; and reinforcement, where repeated exposures strengthen the likelihood of adoption. These enhancements enable cpSI-R to more accurately reflect real-world diffusion processes, where engagement is not necessarily one-time or permanent but often involves multiple phases of activity.The model's formulation ensures monotonicity and submodularity, enabling computationally efficient seed selection strategies with provable guarantees on influence spread while accommodating dynamic re-engagement patterns in evolving networks.

Given the model's ability to capture realistic patterns of influence including persistence, reactivation, and reinforcement, a natural objective is to identify a set of initial seed nodes that can maximize the spread of influence over a given time horizon. This leads to the formal definition of the influence maximization on temporal networks as:

\begin{problem}(Influence Maximization on Temporal Networks)
Given a temporal network $\mathcal{G}$, a diffusion model $\mathcal{M}$, budget $k \in \mathbb{N}$, and spread function $\sigma$, the Influence Maximization (IM) problem is to find a seed set $S_0 \subseteq V$ such that $|S_0| \leq k$ and
\[
S_0^* = \arg\max_{S_0 \subseteq V, |S_0| \leq k} \sigma_{\mathcal{M}}(S_0, \mathcal{G}).
\]
\end{problem}

To address this problem, we present a novel optimization framework that combines a principled seed selection strategy with an adaptive temporal snapshot sampling technique for dynamic networks. The seed selection component leverages the monotonicity and submodularity properties of the spread function to enable efficient approximation under a budget constraint. In parallel, the temporal sampling module adaptively selects representative snapshots by detecting structural shifts in the network, such as changes in edge activity or node interaction patterns, rather than relying on uniform time slicing. This structural awareness ensures that the sampled snapshots preserve both local and global diffusion relevant dynamics, allowing for accurate estimation of influence spread with significantly reduced computational cost.

Our approach not only improves scalability by avoiding redundant computations across temporally similar network states but also enhances the quality of seed selection by focusing on temporally salient moments of influence. Extensive empirical evaluations on several real world temporal datasets, including both medium scale and large scale social contact networks, demonstrate that our framework consistently outperforms existing baseline methods in terms of influence spread and runtime efficiency. In addition, the cpSI-R diffusion model alone exhibits a substantial advantage over traditional models by achieving higher infection spread across temporal snapshots, effectively capturing persistent and reactive behavioral patterns. The improvements offered by both the model and the seed selection strategy become more significant as the network size and temporal resolution increase, highlighting the practical strength of our approach for realistic and large scale diffusion settings.

\subsection{Our Contributions}
Building on the cpSI-R model and the proposed optimization pipeline, we consolidate our technical contributions aimed at advancing influence maximization in temporal networks. Our work not only models influence dynamics with greater behavioral and temporal accuracy, but also introduces efficient algorithms that are both theoretically justified and practically scalable. Below, we summarize the main contributions of this study:

\begin{itemize}
    \item We analyze the role of active-inactive node transitions such as temporary disengagement and reactivation, and demonstrate their substantial impact on the overall influence spread in evolving networks.
    
    \item We propose \textbf{cpSI-R}, a novel diffusion model that integrates both temporal evolution and behavioral aspects, including influence persistence and repeated reinforcement. The resulting influence function is proven to be monotone and submodular, enabling the use of efficient greedy approximation algorithms for seed selection.
    
    \item We develop an optimised influence computation strategy that leverages submodularity to avoid redundant evaluations, thereby significantly improving computational efficiency during seed selection.
    
    \item We design a structure-aware adaptive temporal sampling technique that preserves both local and global temporal dependencies. This sampling enhances the fidelity of diffusion estimation and scales effectively to large temporal networks.
    
    \item We assess the effectiveness of our method using extensive experiments on real-world temporal datasets. Our approach consistently achieves superior influence spread and runtime efficiency compared to leading baseline methods, with especially notable improvements on large-scale networks.
\end{itemize}

\subsection{Organization}
The remainder of this paper is structured as follows: Section \ref{related_work} reviews related work in the domain of influence maximization and temporal networks. Section \ref{preliminaries} outlines the necessary preliminaries. Section \ref{active_inactive} presents the discussion on information diffusion in active-inactive environments. Section \ref{proposed_approach} discusses our proposed approach followed by section \ref{experiments} that discussed experimental setup and results. Finally, Section \ref{discuss} and \ref{conclusion}  concludes the paper and outlines directions for future research.

\section{Related Work}
\label{related_work}
Influence maximization (IM) has been extensively studied since its inception by Kempe et al.~\cite{kempe}, where it was formulated as selecting a set of seed nodes to maximize the expected spread of influence under a chosen diffusion model. The problem is NP-hard, but if the objective function satisfies monotonicity and submodularity, approximation guarantees can be achieved. Monotonicity implies that adding nodes to a seed set does not reduce its influence spread, while submodularity reflects diminishing marginal gains. These properties have been proven for classical models, enabling the use of greedy algorithms that iteratively select nodes with the highest marginal gain~\cite{kempe}. Although exact solutions are computationally intractable, the greedy approach offers a $(1 - 1/e)$ approximation ratio and remains a widely adopted baseline, especially under diffusion models like Independent Cascade (IC), Linear Threshold (LT), and Susceptible-Infected-Recovered (SIR). Unlike these early models, our cpSI-R framework is tailored for temporal networks and designed to retain both monotonicity and submodularity even in the presence of active-inactive transitions, ensuring theoretical guarantees for greedy optimization.

To address the computational challenges of greedy algorithms, various strategies have been developed~\cite{li2018influence}. Simulation-based methods use Monte Carlo simulations for accuracy but at high computational cost~\cite{rw15,rw17}, while proxy-based heuristics improve scalability at the cost of precision~\cite{rw18}\cite{rw19}\cite{rw20}\cite{rw21}\cite{rw22}\cite{rw23}\cite{rw24}\cite{rw25}\cite{rw26}. Sketch-based methods, such as Reverse Influence Sampling (RIS)~\cite{rw27}\cite{rw28}\cite{rw29}\cite{rw30}, provide theoretical guarantees with better efficiency. Simulation-based methods like CELF~\cite{rw15} and CELF++~\cite{rw16} optimize marginal gain computations to reduce simulation calls. Heuristic approaches using degree centrality, PageRank~\cite{rw18}, shortest paths~\cite{rw19}, and Collective Influence (CI)~\cite{rw20} identify influential nodes based on structural indicators, although they often overestimate influence due to overlapping neighborhoods. Enhancements like Degree Discount~\cite{rw21} and Group PageRank~\cite{rw22} refine these estimates by adjusting for redundancy. Meanwhile, sketch-based methods construct random reverse reachable (RR) sets to guide node selection, combining accuracy and scalability. Our approach sidesteps the simulation bottleneck by introducing an optimized forward influence computation strategy that avoids redundant operations and leverages temporal sparsity, outperforming traditional heuristic and sampling-based approaches in dynamic settings.

The choice of diffusion model critically affects the quality and efficiency of IM solutions. Classical models such as IC~\cite{kempe}, LT~\cite{sp4}, and SIR~\cite{kempe} provide the foundational assumptions of discrete influence propagation. However, these models often fall short in capturing complex real-world behaviors, prompting the development of extensions. The SEIR model~\cite{sp5} adds an ``exposed" state to simulate latency, while SCIR~\cite{sp2} captures multi-layered interactions. The irSIR~\cite{sp1} emphasizes recovery mechanisms, and ESIS~\cite{sp6} incorporates emotional contagion. Notably, the Fractional SIR model~\cite{sp3} employs fractional calculus to model memory effects, offering a closer approximation to observed diffusion patterns in realistic settings.In contrast, cpSI-R introduces reinforcement via repeated exposures and time-limited activity phases, providing a unified and realistic abstraction that encapsulates persistence, forgetting, and reactivation without losing computational tractability.

Despite these advances, most IM models and algorithms are restricted to static networks, where node and edge structures are fixed. This assumption is inadequate for real-world systems such as social, communication, and epidemiological networks, which evolve over time. Research in temporal influence maximization is comparatively limited but growing. Some efforts extend static methods like Degree Discount~\cite{rw21}, CI~\cite{rw20}, and RIS~\cite{rw28} to dynamic networks by incorporating time-varying properties. For instance, Dynamic Degree Discount and Dynamic CI leverage evolving node neighborhoods~\cite{imd4}, and temporal RIS adapts RR-set generation across snapshots~\cite{imd1}. These methods improve efficiency but may struggle to maintain high influence spread. Adaptive approaches such as the greedy algorithm under Dynamic Independent Cascade~\cite{imd5} offer theoretical guarantees by adjusting to changing propagation probabilities. Other methods like DIM~\cite{imd2} sample and compress multi-temporal graphs to estimate node reachability, while probing-based strategies approximate seeds with partial network views~\cite{imd6}. Incremental seed updates based on heuristics and upper bounds~\cite{imd3, imd7} further reduce computational costs. A comprehensive taxonomy of these methods is provided by Yanchenko et al.~\cite{yanchenko2024influence}, classifying approaches into single, sequential, maintenance, and probing-based seeding. Our cpSI-R model differs fundamentally from these in its design: it preserves essential theoretical properties in a temporal setting, and our method avoids repeated RR-set constructions or compression schemes, instead directly exploiting temporal patterns and snapshot similarity for sampling efficiency.

Still, influence maximization in temporal networks inherits the NP-hardness of the static case~\cite{kempe}, and the greedy strategy, while extendable~\cite{imd9}, becomes costlier due to repeated simulations and assumptions of full knowledge of the network and model parameters. More critically, diffusion models like SIR, which are monotone and submodular in static graphs, may lose these properties in temporal settings. When node reactivation or recovery (e.g., $\mu > 0$) is introduced, monotonicity and submodularity often break down~\cite{imd9}, weakening the performance of greedy algorithms. Our cpSI-R model overcomes this barrier by enforcing monotonicity and submodularity even with recovery and reactivation, making it the first temporal diffusion model (to our knowledge) that retains these desirable optimization properties.

To better align diffusion modeling with temporal dynamics, researchers have introduced models that incorporate delayed activation, repeated influence attempts, memory, and evolving topologies. Gayraud et al.~\cite{imd9} proposed the Evolving Linear Threshold (ELT) model, including transient (tELT) and persistent (pELT) variants. While tELT considers influence within a single time step, pELT accumulates influence over time and preserves submodularity. Similarly, Kim et al.~\cite{imd10} introduced Continuous-Time IC (CT-IC), relaxing the one-shot activation constraint. Zhu et al.~\cite{msp1} used a Continuous-Time Markov Chain (CTMC-ICM) for finer-grained influence estimates. Zhang et al.~\cite{zhangs} proposed a general cascade model with set-dependent activation probabilities, enhancing realism while maintaining order-independence. Yang et al.~\cite{msp2} introduced t-IC to model multiple activation attempts with timing constraints. Gayraud et al. further developed EIC with transient and persistent versions, where pEIC ensures monotonicity and submodularity under constant probabilities~\cite{imd9}. Hao et al.~\cite{msp3} contributed the Time-Dependent Comprehensive Cascade (TCC) model, adapting activation based on past failures. Aggarwal et al.~\cite{msp4} modeled time-sensitive information retention, and Fu et al.~\cite{imd11} proposed ICEL, which incorporates edge-specific probabilities and dynamic similarity-based updates. cpSI-R generalizes many of these ideas by explicitly modeling (i) the decaying influence over time, (ii) reinforcement via repeated contact, and (iii) time-bounded activation windows—features not simultaneously addressed by prior models. Moreover, we empirically demonstrate that cpSI-R yields significantly higher influence spread in fewer computations, making it both scalable and effective.

To this end, we present a comprehensive framework that bridges the gap between realistic diffusion modeling and scalable optimization for influence maximization in temporal networks. Our proposed model, cpSI-R, incorporates key behavioral and temporal factors such as limited duration activation, cumulative reinforcement through repeated exposures, and the reactivation of previously inactive nodes—dynamics that are often neglected in classical models. To complement this, we introduce an optimized influence maximization algorithm that capitalizes on temporal and structural sparsity to minimize redundant computations, while a structure aware adaptive snapshot sampling technique ensures fidelity in representing network evolution. Collectively, these components yield a principled, efficient, and scalable solution for influence maximization, demonstrating particularly strong performance on large scale and rapidly evolving social networks.

\section{Preliminaries}
\label{preliminaries}

This section defines the fundamental assumptions and notations for influence maximization in dynamic temporal networks. Key concepts such as diffusion models, influence functions, and their properties like monotonicity and submodularity are formally introduced to set the groundwork for the proposed method.

\begin{definition}(Temporal Network)
A temporal network is a dynamic graph $\mathcal{G} = (V, E_T)$, where $V$ is a set of nodes and $E_T \subseteq V \times V \times \mathbb{T}$ is the set of timestamped edges. Each edge $(u, v, t) \in E_T$ represents an interaction between nodes $u$ and $v$ at time $t$, with $\mathbb{T} = \{t_0, t_1, \dots, t_L\}$ denoting the discrete observation time window.
\end{definition}

\begin{definition}(Snapshot Graph and Aggregated Graph)
A snapshot graph at time $t$ is $G_t = (V, E_t)$ where $E_t = \{(u,v) \mid (u,v,t) \in E_T\}$. The time-aggregated graph is defined as $G_{\text{agg}} = (V, E_{\text{agg}})$, where $E_{\text{agg}} = \{(u,v) \mid \exists t : (u,v,t) \in E_T\}$.
\end{definition}

\begin{definition}(Temporal Path)
A temporal path from node $u$ to $v$ is a sequence of timestamped edges $(u_0, u_1, t_1),\\
(u_1, u_2, t_2), \dots, (u_{m-1}, v, t_m)$ such that $t_1 < t_2 < \dots < t_m$. Node $v$ is reachable from $u$ if such a temporal path exists.
\end{definition}

\begin{definition}(Diffusion Model)
A diffusion model $\mathcal{M}$ defines how information (or influence) propagates through the network over time. At each time step, a node may change its state based on rules specific to the model (e.g., Independent Cascade, Linear Threshold).
\end{definition}

\begin{definition}(Seed Set)
The seed set $S_0 \subseteq V$ denotes the initially active or infected nodes at time $t_0$.
\end{definition}

\begin{definition}(Expected Influence Spread)
Given a temporal network $\mathcal{G}$ and a diffusion model $\mathcal{M}$, the expected influence spread of seed set $S_0$ is defined as:
\[
f_{\mathcal{M}}(S_0, \mathcal{G}) = \mathbb{E}_{\mathcal{M}} \left[ \left| \sigma(S_0, \mathcal{G}) \right| \right],
\]
where $\sigma(S_0, \mathcal{G})$ is the random set of nodes activated by the end of the diffusion process.
\end{definition}

\begin{definition}(Monotonicity {\cite{kempe}})
Let $G = (V, E)$ be a social graph. A function $\sigma(S)$ is said to be \emph{monotone increasing} if for any $S \subseteq T \subseteq V$, it holds that:
\[
\sigma(S) \leq \sigma(T).
\]
\end{definition}

\begin{definition}(Submodularity {\cite{kempe}})
Let $G = (V, E)$ be a social graph. A function $\sigma(S)$ is \emph{submodular} if for all $S \subseteq T \subseteq V$ and for any $x \notin T$, the following holds:
\[
\sigma(S \cup \{x\}) - \sigma(S) \geq \sigma(T \cup \{x\}) - \sigma(T).
\]
\end{definition}

\begin{remark}
A node $v \in V$ is said to be \emph{activatable} from $S_0$ under $\mathcal{M}$ if there exists at least one realization of the diffusion process in which $v$ becomes active. Reachability in the temporal graph is a necessary (but not sufficient) condition for activation.
\end{remark}

\begin{table}
  \caption{Notation and Definitions}
  \label{tab:notation}
  \begin{tabular}{cl}
    \toprule
    Notation & Description \\
    \midrule
    $V_1, V_2, \ldots, V_8$ & Node labels used in examples and figures \\
    $\sigma_{\mathcal{M}}(S)$ & Spread function of seed set $S$ under model $\mathcal{M}$ \\
    $S, S'$ & Seed sets for diffusion initiation \\
    $p_{uv}(t_{k}^{uv})$ & Infection probability from node $u$ to $v$ at the $k$-th attempt/time \\
    $\tau$ & Time limit for infection spreading by an infected node \\
    $\delta_u$ & Last infection attempt time by node $u$ \\
    $f(t_k^{uv})$ & Temporal decay function for infection probability \\
    $p_0$ & Base infection probability \\
    $\alpha$ & Reinforcement factor increasing infection probability upon repeated exposure \\
    $\beta, \gamma$ & Scaling and decay parameters in infection probability function \\
    $G(t) = (V(t), E(t))$ & Temporal graph at time $t$ with nodes $V(t)$ and edges $E(t)$ \\
    $L(G_1, G_2)$ & Jaccard similarity between graphs $G_1$ and $G_2$ \\
    $K(G_1, G_2)$ & Kulczynski similarity between graphs $G_1$ and $G_2$ \\
    $S(t)$ & Seed nodes active at time $t$ \\
    $N(t)$ & Neighbor nodes of $S(t)$ at time $t$ \\
    \bottomrule
  \end{tabular}
\end{table}

\section{Information Diffusion with Active-Inactive Transitions}
\label{active_inactive}
 Information diffusion models on temporal networks aim to capture how ideas, behaviors, or diseases propagate through systems where interactions evolve over time \cite{Holme_2012}. These networks allow connections between nodes to appear and disappear, closely reflecting real-world phenomena such as social communications, contact patterns, and online interactions. Standard diffusion models like the Independent Cascade (IC) and Linear Threshold (LT) have been adapted to temporal settings, often including transitions between active and inactive states. However, in many such models, once a node becomes inactive or disengaged, it remains unaffected by future influence attempts—ignoring the reality that individuals can re-engage under sustained social pressure or repeated exposures. This modeling limitation overlooks key behavioral dynamics such as temporary disengagement, reactivation through peer influence, and reinforcement effects that occur in real-world diffusion processes.

To overcome these challenges, we introduce the novel \textit{cpSI-R} model, which integrates temporal evolution with behavioral persistence and reactivation mechanisms. By explicitly modeling how nodes can re-enter the active state after disengagement and how influence accumulates through reinforcement, cpSI-R offers a more realistic and expressive diffusion framework. As a result of these structured transitions, the model naturally recovers monotonicity and submodularity in the spread function, enabling efficient optimization without compromising behavioral fidelity.We begin by formally defining the concept of active-inactive transitions that form the foundation of our modeling framework,and discuss the complete model in the next section.

\begin{definition}
\textbf{(Active-Inactive Transitions in Temporal Diffusion Models)}  
Let \( \mathcal{T} = \{1, 2, \ldots, T\} \) denote the set of discrete time steps, and let \( \mathcal{G} = \{G_t = (V_t, E_t) \mid t \in \mathcal{T} \} \) represent a temporal network. Here, each \( G_t \) is a graph snapshot at time \( t \), where \( V_t \subseteq V \) denotes the set of nodes present at that time, and \( E_t \subseteq V_t \times V_t \) represents the set of active edges at time \( t \). The set \( V \) is the union of all node sets over time, i.e., \( V = \bigcup_{t \in \mathcal{T}} V_t \). Each node \( v \in V \) is associated with a binary state function \( s_v(t) \in \{0, 1\} \), where \( s_v(t) = 1 \) indicates that node \( v \) is active at time \( t \in \mathcal{T} \), and \( s_v(t) = 0 \) indicates that it is inactive. In temporal diffusion models that incorporate active-inactive transitions, a node can switch from the active to the inactive state after some time, typically due to lack of sufficient influence from its neighbors or based on specific rules defined by the model. A common assumption in irreversible diffusion frameworks is that once a node becomes inactive, it cannot become active again. Formally, if a node \( v \in V \) is active at time \( t_1 \in \mathcal{T} \) and transitions to the inactive state at a later time \( t_2 \in \mathcal{T} \), where \( t_1 < t_2 \), then it remains inactive at all subsequent time steps. This can be expressed as:
\[
s_v(t_1) = 1 \quad \text{and} \quad s_v(t_2) = 0 \quad \Rightarrow \quad s_v(t) = 0 \quad \forall \, t > t_2.
\]
\end{definition}

This implies that once a node transitions to the inactive state, it cannot return to the active state, thus making its influence spread temporary. As a result, the influence \( \sigma(S) \) of a seed set \( S \subseteq V \) becomes non-monotonic; adding more nodes to \( S \) does not necessarily increase the total influence spread because the duration of influence is constrained by the active state period of each node.

\begin{lemma}
The spread function $\sigma_D(S)$ is neither monotone nor submodular under any diffusion model that allows for only active-inactive transition.
\end{lemma}

\begin{proof}
We provide a counterexample as shown in Figure 1 for which the spread function is neither monotone nor submodular if there is a window for an activated node to become inactivated. Assume the probability of activation of a node \( b \) by node \( a \) at time \( t \) is 1, i.e., \( p_{ab}^t = 1 \). In this example, we have a temporal network \( T \) represented by snapshots \( T^1, T^2, T^3, T^4 \).

Figure 1 shows the diffusion process when the seed set is \( S = \{V_1^0\} \). Triangle-shaped nodes represent active or infected nodes at a given time \( t \), round nodes represent those susceptible to infection in future time stamps greater than \( t \), and rectangular nodes represent nodes that have undergone an active-inactive transition. The spread function \( \sigma_D(S) \) calculates the number of nodes infected through the chain of activations until the last snapshot \( T^4 \). The influence obtained by the spread function equals the total number of rectangular and triangular nodes in \( T^4 \), yielding a resulting spread of \( \sigma_D(\{V_1^0\}) = 6 \).

Now, consider the addition of node \( V_3^0 \) to the seed set \( S \). Figure \ref{f1} shows the diffusion process. The activation of node \( V_3 \) at time \( t = 0 \) causes node \( V_4 \) to be activated at \( t = 1 \) (see figure \ref{f1}(a), followed by the active-inactive transition at \( t = 2 \), blocking the spread of infection to \( V_5 \) and other nodes connected by the path \( V_4V_5 \) (see figure \ref{f1}(b). This premature activation of \( V_3 \), leading to the early active-inactive transition of node \( V_4 \), reduces the influence spread. As a result, the spread becomes \( \sigma_D(\{V_1^0, V_3^0\}) = 4 < \sigma_D(\{V_1^0\}) \), proving that \( \sigma_D(S) \) is non-monotone.

We use a similar approach to demonstrate that \( \sigma_D \) is not submodular. Consider two sets \( A = \{V_1^0\} \) and \( B = \{V_1^0, V_3^0\} \), where \( A \subset B \). Adding node \( V_8^0 \) to set \( A \) results in a total influence of 7, as shown in Figure \ref{f2}. Adding the same node to set \( B \) yields a total spread of 8 (see Figure \ref{f2}(b)), i.e., \( \sigma_D(B \cup \{V_8^0\}) = 8 \). Thus, we have:
\[
\sigma(A \cup \{V_8^0\}) - \sigma(A) = 7 - 6 = 1 \quad \text{and} \quad \sigma(B \cup \{V_8^0\}) - \sigma(B) = 8 - 4 = 4.
\]
Since \( A \subset B \), it follows that:
\[
\sigma(A \cup x) - \sigma(A) \not\geq \sigma(B \cup x) - \sigma(B),
\]
which violates the submodularity property.

We conclude that the spread function \( \sigma_D(S) \) is neither monotone nor submodular in temporal networks where nodes exhibit recovery probability. This implies that there is no guarantee on the optimality gap for the solution obtained using a greedy algorithm under the IC, LT, or SIR models. However, for the models that dont allow complete recovery, we can still achieve a solution that is at most \( 1 - \frac{1}{e} \) times away from the optimal solution by applying a greedy method \cite{kempe}. Additionally, computation costs can be reduced using lazy forward optimization \cite{celf++}. In the following subsection, we present an SI-like information diffusion model called the Continuous Persistent SI model with Reinforcement (cpSI-R).
\end{proof}
\begin{figure}[ht]
    \centering
    \subfigure[]{\includegraphics[width=0.45\textwidth]{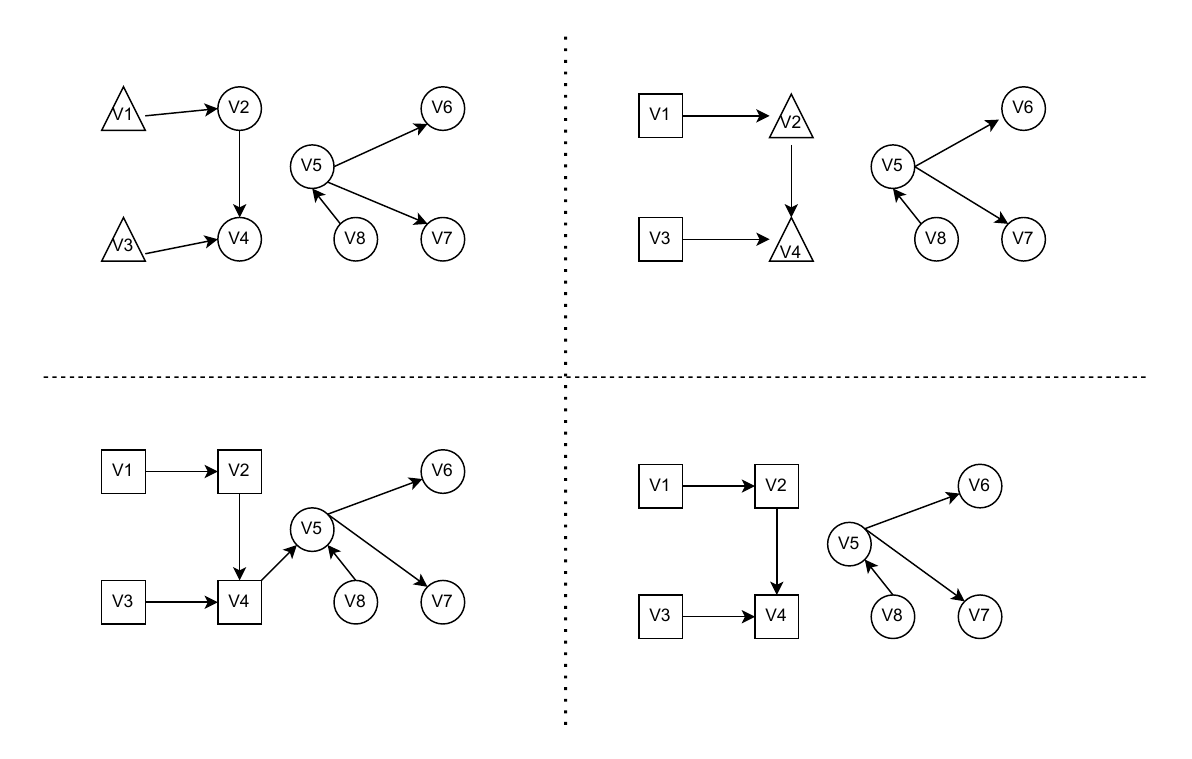}} 
    \subfigure[]{\includegraphics[width=0.45\textwidth]{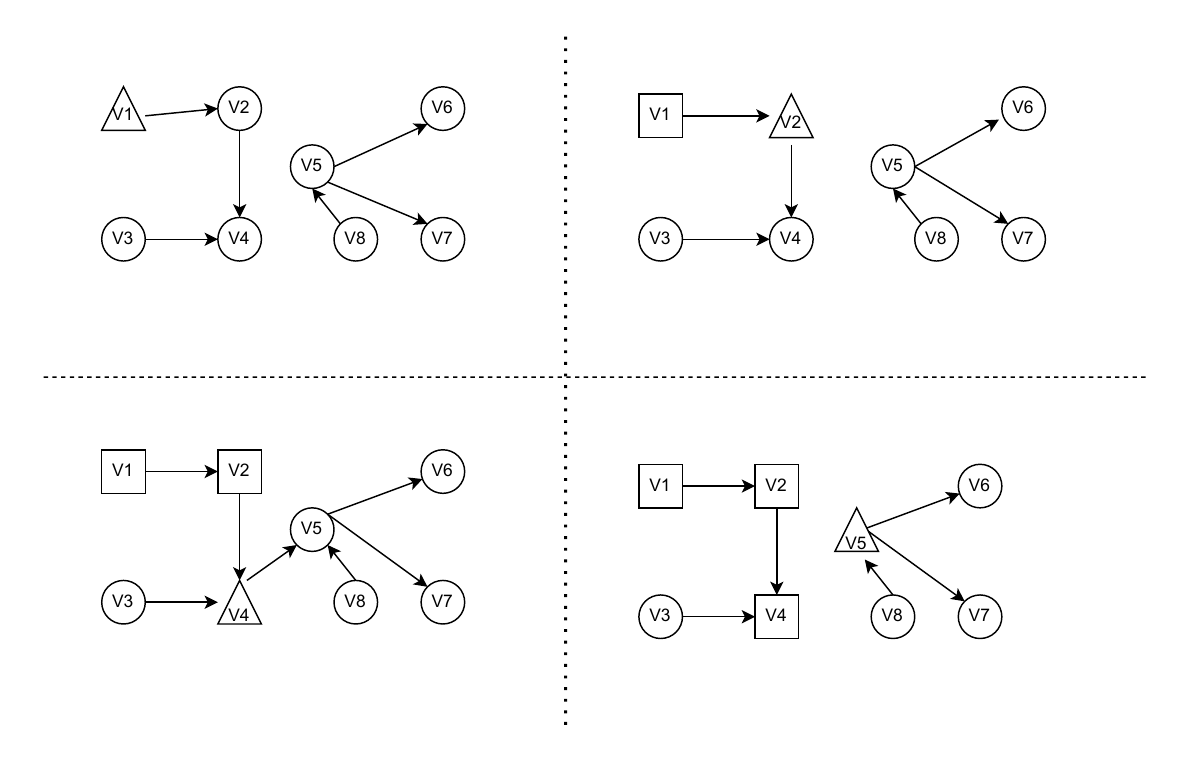}} 
    \caption{Counter example showing the violation of monotonicity under active-inactive transition (a): Seed set contains two infected nodes $V_1$ and $V_3$  (b): Seed set contains a single infection $V_1$}
    \label{f1}
\end{figure}

\begin{figure}[ht]
    \centering
    \subfigure[]{\includegraphics[width=0.45\textwidth]{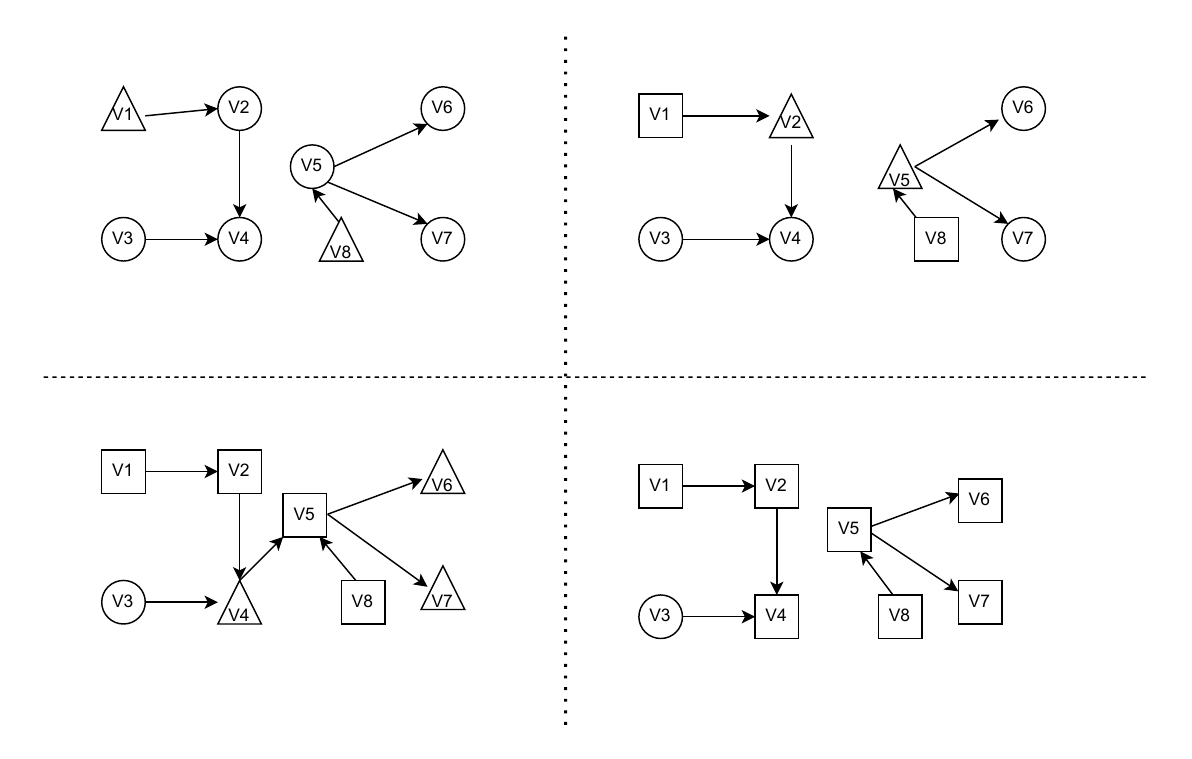}} 
    \subfigure[]{\includegraphics[width=0.45\textwidth]{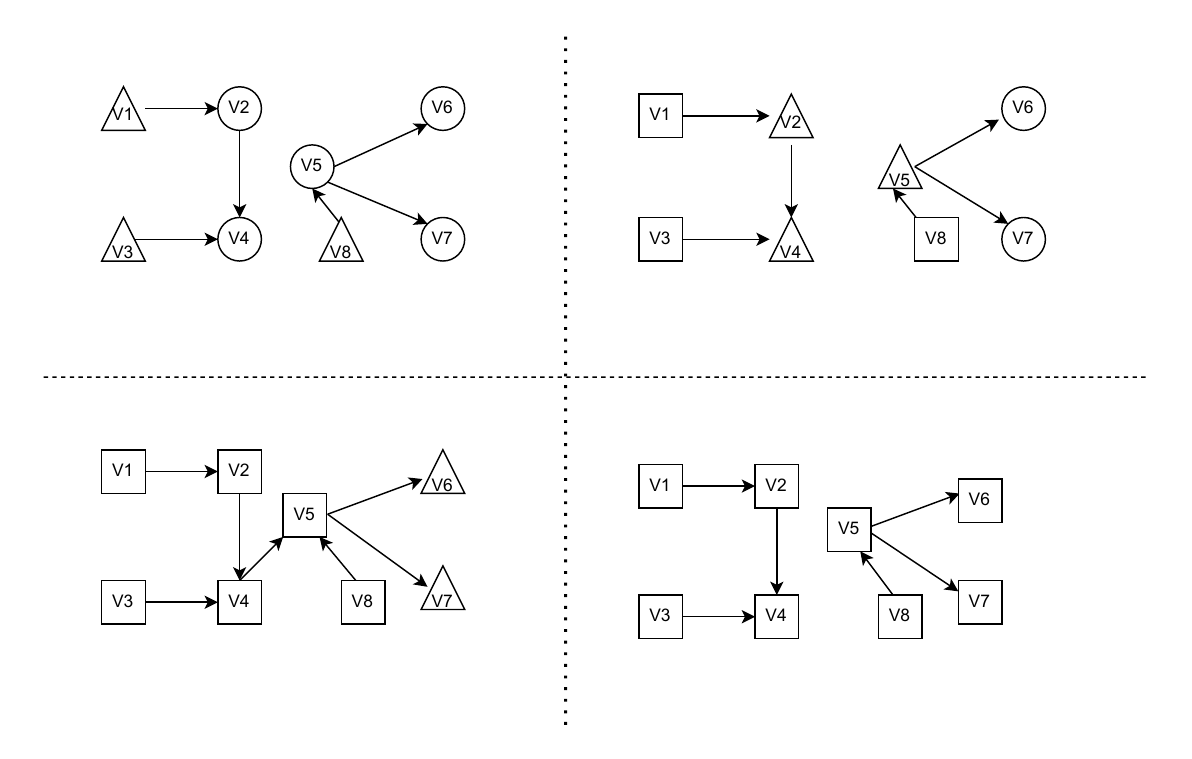}} 
    \caption{Counter example showing the violation of submodularity under active-inactive transition (a): Seed set contains one infected node $V_1$ and an additional node $V_8$ (b): Seed set contains $V_1 , V_3$ and an additional node $V_8$}
    \label{f2}
\end{figure}

\newpage
\section{Proposed Approach}
\label{proposed_approach}
Real-world diffusion processes are inherently temporal and influenced by behavioral factors such as reinforcement, relapse, and persistence. Traditional influence maximization (IM) methods on static or temporal networks often fail to capture these dynamics, limiting their applicability in real-life scenarios. To address these challenges, we propose a novel influence diffusion model—Continuous Persistent Susceptible-Infected with Reinforcement and Reactivation (cpSI-R) that operates over temporal networks while accounting for key temporal and behavioral characteristics of information spread.

The cpSI-R model enhances the classical Susceptible-Infected framework by introducing time-limited reinforcement and reactivation mechanisms, allowing previously influenced (infected) nodes to regain influence capability based on recurring exposure. This continuous-time model also ensures that influence is not transient but persistent over realistic time windows, enabling a more accurate reflection of real-world influence behavior.

To effectively leverage the cpSI-R model, we formulate an objective function that captures the expected cumulative influence over a time horizon, and prove that it satisfies monotonicity and submodularity, making it suitable for efficient approximation. We further design an optimized seed selection strategy that discretizes the temporal network intelligently, eliminates redundant computations, and incorporates lazy forward evaluation to significantly improve computational efficiency.

The rest of this section begins with an overview of information diffusion in temporal networks, followed by a formal description of the cpSI-R model. It then details the construction and properties of the objective function, including proofs of its monotonicity and submodularity. Subsequently, the strategy for influence maximization using the cpSI-R model is presented, along with techniques to reduce computational costs. Finally, a lazy forward influence strategy is introduced to accelerate seed selection while preserving theoretical guarantees.

\subsection{cpSI-R Model}

In traditional diffusion models, influence spreads in a simplistic, often static manner, failing to capture the dynamic and temporal nuances of real-world interactions. Temporal social networks, however, are characterized by active-inactive transitions, where nodes do not continuously exert influence. Instead, their influence can fluctuate based on interaction patterns, time, and repeated exposures. These temporal dynamics are crucial in scenarios like behavioral change campaigns or viral marketing, where repeated interactions increase the likelihood of adoption but may diminish over time if not reinforced. Existing models lack mechanisms to simulate such persistent yet time-limited influence, leading to suboptimal or unrealistic predictions in dynamic networks.

To address this gap, we introduce the Continuous Persistent Susceptible-Infected Model with Time-Limited Reinforcement and Re-activation (cpSI-R), specifically designed to reflect the persistence and decay of influence in temporal settings. This model describes a diffusion process where nodes (users) that become infected (adopt a behavior) continue to try to infect (influence) their neighbors within a designated time frame. If an infected node fails to spread the infection within this period, it loses its ability to influence others unless reactivated by interacting with another active node. The infection probability increases with each subsequent exposure, mirroring real-world situations where repeated interactions strengthen the likelihood of adoption.

A maximum time frame \( \tau \) is introduced within which an infected node can spread the infection. If an infected node does not spread the infection within \( \tau \) time units, it becomes inactive for spreading influence until reactivated by an active node. Let \( \delta_u \) denote the time of the last infection attempt by node \( u \). If \( t_{uv}^k - \delta_u > \tau \), then node \( u \) becomes inactive for spreading the infection until reactivated. This captures scenarios like a social media fitness campaign, where users may be encouraged to adopt healthy habits through continuous exposure to fitness-related posts. Each time a user sees a fitness post, their likelihood of adopting the habit increases, with influence fluctuating over time based on interaction patterns. If users disengage from fitness content for an extended period, they lose the drive to influence others unless re-engaged by active users.

When a node \( u \) becomes infected at time \( t \), it attempts to influence a susceptible node \( v \) based on a series of temporal exposures, represented by the sequence \( \{t_{uv}^k\}_{k \geq 1} \) where an edge between \( u \) and \( v \) exists after time \( t \). Each exposure increases the infection probability \( p_{uv}(t_{uv}^k) \):
\[
p_{uv}(t_{uv}^k) = p_0 \cdot (1 - e^{-\alpha k}) \cdot f(t_{uv}^k),
\]
where \( p_0 \) is the base infection probability, \( \alpha \) is a reinforcement factor, \( k \) is the exposure count, and \( f(t_{uv}^k) \) is a temporal function adjusting the probability based on \( t_{uv}^k \). At each \( t_{uv}^k \), if \( v \) is susceptible, \( u \) tries to infect \( v \) with probability \( p_{uv}(t_{uv}^k) \). If successful, \( v \) becomes infected, continuing the process from \( v \).

The time-dependent influence is modeled by:
\[
f(t_{uv}^k) = \beta e^{-\gamma (t_{uv}^k - t)},
\]
where \( \beta \) is a scaling factor and \( \gamma \) controls the decay rate of influence over time.

\begin{theorem}
For any instance of the cpSI-R model on a temporal graph with fixed probabilities and a temporal interaction function, the spread function \( \sigma_{cpSI\text{-}R} \) is monotone and submodular.
\end{theorem}

\begin{proof}
Let \( V \) denote the set of nodes in the temporal graph \( \mathcal{G} = \{G_t = (V_t, E_t)\}_{t \in \mathcal{T}} \), where \( \mathcal{T} = \{1, 2, \ldots, T\} \) represents the set of discrete time steps. The influence spread function \( \sigma_{cpSI\text{-}R} : 2^V \rightarrow \mathbb{R}_{\geq 0} \) maps any seed set \( S \subseteq V \) to the expected number of nodes that become active during the diffusion process governed by the cpSI-R model.

Let \( \Omega \) denote the sample space of all possible realizations \( \omega \) of the stochastic diffusion process. Each realization \( \omega \in \Omega \) corresponds to a deterministic unfolding of the cpSI-R process, based on fixed samples of random edge activations, temporal interaction weights, reinforcement effects, and reactivation events. Given a realization \( \omega \), the spread function becomes deterministic and is denoted \( \sigma_\omega(S) \), representing the number of nodes that are active at the end of the process starting from seed set \( S \).

The expected influence spread under cpSI-R is given by
\[
\sigma_{cpSI\text{-}R}(S) = \mathbb{E}_{\omega \sim \Omega}\left[ \sigma_\omega(S) \right] = \sum_{\omega \in \Omega} \Pr[\omega] \cdot \sigma_\omega(S).
\]

We now prove that \( \sigma_{cpSI\text{-}R} \) is both monotone and submodular.

Monotonicity: Let \( S \subseteq V \) be a seed set and \( u \in V \setminus S \) be a node not in \( S \). Consider the two sets \( S \) and \( S \cup \{u\} \).

Fix any realization \( \omega \in \Omega \). The value \( \sigma_\omega(S) \) represents the number of nodes activated when the diffusion process starts with \( S \), and \( \sigma_\omega(S \cup \{u\}) \) represents the number of nodes activated when starting from \( S \cup \{u\} \).

In the cpSI-R model, each active node can attempt to activate its neighbors within a limited time window, and reinforcement or reactivation can increase the chance of further activations. Adding node \( u \) to the initial seed set gives it an independent opportunity to activate nodes through its own influence path. Since activations are never revoked and additional seeds only increase influence, it follows that
\[
\sigma_\omega(S \cup \{u\}) \geq \sigma_\omega(S), \quad \text{for all } \omega \in \Omega.
\]

Taking expectation over all realizations, we obtain
\[
\sigma_{cpSI\text{-}R}(S \cup \{u\}) 
= \mathbb{E}_\omega\left[ \sigma_\omega(S \cup \{u\}) \right] 
= \sum_{\omega \in \Omega} \Pr[\omega] \cdot \sigma_\omega(S \cup \{u\}) 
\geq \sum_{\omega \in \Omega} \Pr[\omega] \cdot \sigma_\omega(S) 
= \sigma_{cpSI\text{-}R}(S).
\]

Therefore, \( \sigma_{cpSI\text{-}R} \) is monotone.

Submodularity: Let \( S \subseteq S' \subseteq V \), and let \( u \in V \setminus S' \). Fix a realization \( \omega \in \Omega \). Consider the marginal gain of adding \( u \) to \( S \) and to \( S' \). These are computed as
\[
\sigma_\omega(S \cup \{u\}) - \sigma_\omega(S)
\quad \text{and} \quad
\sigma_\omega(S' \cup \{u\}) - \sigma_\omega(S').
\]

Since \( S \subseteq S' \), the diffusion process starting from \( S' \) activates at least as many nodes as the process from \( S \), possibly more. Some of the nodes that \( u \) could reach when added to \( S \) might already be activated when \( u \) is added to the larger set \( S' \). Therefore, the marginal gain of adding \( u \) to \( S \) is greater than or equal to the gain of adding \( u \) to \( S' \), that is,
\[
\sigma_\omega(S \cup \{u\}) - \sigma_\omega(S) \geq \sigma_\omega(S' \cup \{u\}) - \sigma_\omega(S'), \quad \text{for all } \omega \in \Omega.
\]

Taking expectations on both sides,
\[
\sigma_{cpSI\text{-}R}(S \cup \{u\}) - \sigma_{cpSI\text{-}R}(S)
= \mathbb{E}_\omega\left[ \sigma_\omega(S \cup \{u\}) - \sigma_\omega(S) \right]
\geq \mathbb{E}_\omega\left[ \sigma_\omega(S' \cup \{u\}) - \sigma_\omega(S') \right]
= \sigma_{cpSI\text{-}R}(S' \cup \{u\}) - \sigma_{cpSI\text{-}R}(S').
\]

Therefore, \( \sigma_{cpSI\text{-}R} \) is submodular.

Since the function is both monotone and submodular, the proof is complete.
\end{proof}

In our cpSI-R  model, we extend the dynamics to incorporate time-limited interactions where nodes can transmit information only within specific time frames unless reactivated. The model captures diffusion dynamics with increasing infection probabilities due to repeated exposures, reflecting scenarios where repeated interactions or reinforcements increase the likelihood of adoption.Let's denote the node set of the network at time \( t \) as \( V(t) \) and the edge set as \( E(t) \), forming the dynamic graph \( G(t) = (V(t), E(t)) \).

Each edge \( (u,v) \) in \( E(t) \) has a time-limited transmission probability \( p_{uv}(t_1, t_2) \) over the interval \( (t_1, t_2) \), where \( t_1 \) and \( t_2 \) are times within which the edge \( (u,v) \) is active. This probability is computed based on the periods \( \delta t_m \) when the edge is active.

\subsection{Influence Maximization through cpSI-R}

In the cpSI-R model, accurately assessing node influence over time intervals is essential to understanding diffusion dynamics. We introduce a method that discretizes time intervals based on both global and local structural similarities. This approach ensures that the selected intervals align with significant structural changes in the network, thereby enhancing the fidelity of influence estimation. 

Our contribution lies in refining the dynamic update equations within the cpSI-R framework and reducing computational costs by avoiding redundant influence probability computations for seed nodes across multiple iterations. These equations incorporate time-dependent transmission probabilities, derived through a sampling algorithm that adapts to both global and local structural evolution. This adaptation not only improves computational efficiency but also enhances the model's accuracy in representing real-world diffusion processes where network structures evolve dynamically. Furthermore, by ensuring that the sampling captures key moments of topological change, the model becomes well-suited for applications such as epidemic control, targeted marketing, and information dissemination in rapidly changing environments. This nuanced understanding of dynamic influence provides decision-makers with precise tools for optimizing interventions in temporal settings.

\begin{algorithm}
\caption{TemporalInfluenceMaximization}
\begin{algorithmic}[1]
\Require Transmission Matrix $F(\cdot)$, Time Horizon $(t_0, t_0 + h)$, Seed Set Size $k$, Temporal Network $T$, Maximum Time Stamp $\text{max\_t}$, Threshold $\eta$, Weighting Factors $\alpha, \beta$, Minimum Iterations $\text{MinIter}$
\Ensure Set $S$ of $k$ influential nodes

\State $t_{\text{set}} \gets \text{SamplingAlgorithm}(T, \text{max\_t}, \eta, \alpha, \beta)$ \Comment{Step 1: Temporal Sampling}
\For{each node $i \in N(t_0)$}
    \State $\text{CalcInfluence}(\{i\}, F(\cdot), (t_0, t_0 + h))$ \Comment{Step 2: Initialize Influence Calculation}
\EndFor

\State $S \gets \text{LazyForwardInfluence}(F(\cdot), (t_0, t_0 + h), k, \text{MinIter})$ \Comment{Step 3: Lazy Forward Influence Maximization}

\State $\text{FinalInfluence} \gets \text{CalcInfluence}(S, F(\cdot), (t_0, t_0 + h))$ \Comment{Step 4: Final Influence Calculation for Selected Set $S$}

\State \Return $S$
\end{algorithmic}
\end{algorithm}

\begin{algorithm}
\caption{CalcInfluence}
\begin{algorithmic}[1]
\Require Seed Set $S$, Transmission Matrix $F(\cdot)$, Time Stamp $(t_0, t_0 + h)$
\Ensure Influence of seed set $S$
\State Initialize $p(i, t_0) = 1$ for each $i \in S$ and $0$ otherwise
\State Divide $(t_0, t_0 + h)$ into time periods $t_0, \ldots, t_r = t_0 + h$ using \textit{SamplingAlgorithm}
\State $j \gets 0$
\Repeat
    \State Compute edge set $E(t_j, t_{j+1})$ and spread matrix $P(t_j, t_{j+1})$
    \For{each $i$}
        \State Compute $p(i, t_{j+1})$ using dynamic update equations for cpSI-R model
    \EndFor
    \State $j \gets j + 1$
\Until{$j = r$}
\State \Return $|S| + \sum_{i \in N(t) / S(t)} p(i, t_r)$
\end{algorithmic}
\end{algorithm}

\texttt{Algorithm 2} initializes node influence probabilities based on the given seed set and discretizes the interval $(t_0, t_0 + h)$ using the Sampling Algorithm that dynamically adjusts to structural changes in the network. It iteratively computes transmission probabilities and updates influence estimates using the modified dynamic equations designed for the cpSI-R model. By focusing on only the actively influencing nodes and using an efficient sampling approach, the algorithm minimizes computational overhead while ensuring accurate influence estimation.The computational cost is further reduced through effective temporal sampling. Oversampling increases redundancy, while undersampling may degrade the quality of selected seeds. We address this by providing an effective method to sample the temporal network. The goal of the sampling algorithm is to partition the overall time window into smaller intervals that meaningfully capture changes in network structure. This step is crucial to understanding how the network evolves over time. Our sampling algorithm uses two complementary similarity measures between consecutive graph snapshots: Jaccard similarity \( L(G_1, G_2) \) and Kulczynski similarity \( K(G_1, G_2) \).

We select these two measures to capture both local and global structural changes. Jaccard similarity measures overlap in edge sets, effectively tracking local, immediate changes in connections. Kulczynski similarity provides a more global perspective by accounting for the proportion of shared edges with respect to the total edges in each snapshot. Together, these measures provide a balanced assessment of structural evolution in the temporal network. The combined score computed from these similarities helps to identify timestamps where significant structural changes occur—key moments for accurate influence estimation in the cpSI-R model.

\begin{definition}[Temporal Jaccard Similarity]
Let \( \mathcal{G} = \{G_t = (V_t, E_t)\}_{t=1}^{T} \) be a temporal network represented as a sequence of graph snapshots over discrete time steps. The Jaccard similarity between two snapshots \( G_{t_1} \) and \( G_{t_2} \), with \( t_1 < t_2 \), is defined as:
\begin{equation*}
L(G_{t_1}, G_{t_2}) = \frac{|E_{t_1} \cap E_{t_2}|}{|E_{t_1} \cup E_{t_2}|},
\end{equation*}
where \( E_t \) denotes the set of active edges at time step \( t \). This measure captures the structural overlap of interactions between two time points in the temporal network.
\end{definition}

\begin{definition}[Temporal Kulczynski Similarity]
Given two snapshots \( G_{t_1} = (V_{t_1}, E_{t_1}) \) and \( G_{t_2} = (V_{t_2}, E_{t_2}) \) from a temporal network \( \mathcal{G} = \{G_t\}_{t=1}^T \), the Kulczynski similarity between them is defined as:
\begin{equation*}
K(G_{t_1}, G_{t_2}) = \frac{1}{2} \left( 
\frac{|E_{t_1} \cap E_{t_2}|}{|E_{t_1}|} + 
\frac{|E_{t_1} \cap E_{t_2}|}{|E_{t_2}|} 
\right),
\end{equation*}
where \( E_{t_1} \) and \( E_{t_2} \) represent the sets of active edges at times \( t_1 \) and \( t_2 \), respectively. This metric quantifies the temporal consistency of edge activity between two time points.
\end{definition}

The cumulative score \( Score \) combines both similarity measures to determine the significance of a timestamp \( t \) for partitioning:
\begin{equation*}
Score = \alpha \cdot L(G_1, G_2) + \beta \cdot K(G_1, G_2)
\end{equation*}
where \( \alpha \) and \( \beta \) are weighting factors controlling the relative importance of Jaccard and Kulczynski similarities.

While efficient influence estimation is important, it is equally critical to determine the specific time intervals where the network undergoes meaningful structural change. These intervals help improve estimation accuracy while keeping computational cost low. The Sampling Algorithm plays a vital role in achieving this balance.

\subsubsection{Adaptive Sampling Based on Structural Dynamics}

To reduce computational overhead while preserving the temporal fidelity of diffusion processes, we propose an adaptive sampling strategy that selects structurally informative time stamps. The key idea is to identify points in time where the network undergoes meaningful structural change, thereby avoiding redundant or low-variance snapshots.

At the core of our method lies a similarity scoring mechanism that assesses the structural evolution between consecutive graph snapshots. Specifically, for each pair of successive snapshots \( G_{t} \) and \( G_{t+1} \), we compute a composite similarity score, denoted as \( \textit{Score} \), which reflects the structural overlap between the two graphs. This score is computed using a convex combination of two widely-used structural similarity metrics: Jaccard similarity and Kulczynski similarity:
\[
\textit{Score}(G_t, G_{t+1}) = \alpha \cdot L(G_t, G_{t+1}) + \beta \cdot K(G_t, G_{t+1}),
\]
where \( L(\cdot, \cdot) \) and \( K(\cdot, \cdot) \) denote the Jaccard and Kulczynski similarities, respectively, and \( \alpha, \beta \in [0,1] \) are their corresponding weights such that \( \alpha + \beta = 1 \). 

The values of \( \alpha \) and \( \beta \) are set based on empirical validation to balance global and local structural perspectives: Jaccard emphasizes global edge set overlap, while Kulczynski adjusts for asymmetric distributions of edges between snapshots. In our experiments, we set \( \alpha = 0.5 \) and \( \beta = 0.5 \), providing equal importance to both similarity aspects, although these can be tuned for specific application contexts.

If the computed \( \textit{Score} \) exceeds a predefined threshold \( \eta \), the corresponding time stamp \( t \) is added to the sampled time set \( t\_set \). Otherwise, the algorithm adaptively increments \( t \) using an exponential step size until a sufficiently distinct structural shift is observed. This dynamic adjustment ensures that only non-redundant, structurally informative snapshots are selected—balancing computational efficiency with the need to preserve critical temporal dynamics. This sampling strategy is designed to complement the \textit{Calcinfluence()} routine, enabling accurate estimation of influence spread by focusing on structurally relevant intervals. By identifying time windows where significant structural transitions occur, this approach facilitates robust analysis of temporal networks and improves the quality of seed influence estimation across time.

Although the sampling algorithm effectively isolates key intervals, identifying the most influential nodes within those intervals is crucial for understanding how influence propagates. The next section addresses this through the introduction of the Lazy Forward Influence algorithm, which performs efficient seed selection within the sampled temporal framework.

\begin{algorithm}[H]
\caption{SamplingAlgorithm}
\begin{algorithmic}[1]
\Require Temporal Network $T$, Maximum Time Stamp $\text{max\_t}$, Threshold $\eta$, Weighting Factors $\alpha, \beta$
\Ensure Set of selected time stamps $t_{\text{set}}$
\State Initialize empty set $t_{\text{set}}$
\State $t \gets 0$
\While{$t < \text{max\_t}$}
    \State $G_1 \gets \text{ObtainSnapshot}(T, t)$
    \State $G_2 \gets \text{ObtainSnapshot}(T, t + 1)$
    \State Compute $L \gets L(G_1, G_2)$ \Comment{Jaccard similarity}
    \State Compute $K \gets K(G_1, G_2)$ \Comment{Kulczynski similarity}
    \State Compute $Score \gets \alpha \cdot L + \beta \cdot K$
    \If{$Score \geq \eta$}
        \State Add $t$ to $t_{\text{set}}$
        \State $t \gets t + 1$
    \Else
        \State $step \gets 1$
        \While{$Score < \eta$ \textbf{and} $t + step \leq \text{max\_t}$}
            \State $t \gets t + step$
            \State $G_2 \gets \text{ObtainSnapshot}(T, t)$
            \State Compute $L \gets L(G_1, G_2)$
            \State Compute $K \gets K(G_1, G_2)$
            \State Compute $Score \gets \alpha \cdot L + \beta \cdot K$
            \State $step \gets 2 \cdot step$
        \EndWhile
        \If{$Score \geq \eta$}
            \State Add $t$ to $t_{\text{set}}$
            \State $t \gets t + 1$
        \Else
            \State $t \gets t + step$
        \EndIf
    \EndIf
\EndWhile
\State \Return $t_{\text{set}}$
\end{algorithmic}
\end{algorithm}

\subsubsection{Lazy Forward Influence}
We introduce the \textit{LazyForwardInfluence} algorithm, which leverages lazy forward optimization \cite{celf} to identify influential nodes in a temporal network. This method eliminates the exhaustive step of evaluating each node in \(N(t_0) - S\) for potential replacement with any node in \(S\). Instead, it incorporates a mechanism to track iterations without replacements and terminates early if no replacements occur for a specified number of iterations. This approach significantly reduces computational overhead while preserving seed quality, ensuring an efficient and effective analysis of temporal network dynamics.

The \textit{LazyForwardInfluence} algorithm  efficiently identify a set of $k$ influential nodes in temporal networks over a specified time horizon. It leverages precomputed influence estimates to construct an initial seed set and employs a priority-based replacement mechanism to iteratively refine this set. At each iteration, it considers the node with the highest potential influence gain from the candidate pool and evaluates its impact when substituted for each node in the current seed set. A replacement is executed only if a substantial gain in overall influence is observed.

The algorithm incorporates a lazy evaluation scheme by maintaining a priority queue of candidates and halting the process early when no beneficial replacements are found over a fixed number of iterations. This approach significantly reduces redundant computations while preserving the influence quality of the selected set. Overall, \textit{LazyForwardInfluence} balances optimization accuracy and computational efficiency, making it suitable for influence maximization in large-scale or evolving temporal networks.

\begin{algorithm}[H] 
\caption{LazyForwardInfluence}
\begin{algorithmic}[1]
\Require Transmission Matrix $F(\cdot)$, Time Horizon $(t_0, t_0 + h)$, Number of Influence Points $k$, Minimum Iterations $\text{MinIter}$
\Ensure Set $S$ of $k$ influential nodes
\For{each node $i \in N(t_0)$}
    \State Compute $\text{CalcInfluence}(\{i\}, F(\cdot), (t_0, t_0 + h))$
\EndFor
\State $S \gets$ initial set of $k$ nodes with the highest $\text{CalcInfluence}$ values
\State $\text{PriorityQueue} \gets$ initialize max-heap with nodes in $N(t_0) - S$ ordered by $\text{CalcInfluence}$ value
\State $\text{NoReplace} \gets 0$
\While{$\text{NoReplace} < \text{MinIter}$ and $\text{PriorityQueue.isEmpty()} = \text{false}$}
    \State $i \gets \text{PriorityQueue.extractMax()}$
    \State $\text{maxIncrease} \gets 0$
    \State $\text{replaceNode} \gets \text{null}$
    \For{each node $j \in S$}
        \State $\text{influenceGain} \gets \text{CalcInfluence}(S \cup \{i\} - \{j\}, F(\cdot), (t_0, t_0 + h)) - \text{CalcInfluence}(S, F(\cdot), (t_0, t_0 + h))$
        \If{$\text{influenceGain} > \text{maxIncrease}$}
            \State $\text{maxIncrease} \gets \text{influenceGain}$
            \State $\text{replaceNode} \gets j$
        \EndIf
    \EndFor
    \If{$\text{replaceNode} \neq \text{null}$ and $\text{maxIncrease} > 0$}
        \State $S \gets S \cup \{i\} - \{\text{replaceNode}\}$
        \State $\text{NoReplace} \gets 0$
    \Else
        \State $\text{NoReplace} \gets \text{NoReplace} + 1$
    \EndIf
\EndWhile
\State \Return $S$
\end{algorithmic}
\end{algorithm}

With the mechanisms for identifying key time intervals and influential nodes in place, it is important to illustrate the practical utility of these algorithms. A toy example on a temporal network will demonstrate the effectiveness of the cpSI-R model and the associated algorithms in capturing influence dynamics over time.

\textbf{Example:} To illustrate the working of the cpSI-R model and the algorithms, we assume a toy temporal network as shown in Figure 3.

We begin by allowing the network to evolve in one direction, where edges are added over time, while deletions are restricted for simplicity. This ensures that we study the diffusion process in a setting where network complexity gradually increases. A key component of our approach is the sampling algorithm, which partitions time into intervals based on structural similarity between consecutive snapshots. We chose \( \eta = 0.3 \), which leads to the selection of snapshots 1, 2, 4, and 8 based on their similarity scores, drastically reducing computational time. For example, if selecting all snapshots required \( O(T) \) computations where \( T \) is the total number of snapshots, selecting fewer snapshots based on \( \eta \) reduces this to \( O(\eta T) \), offering computational efficiency, particularly when \( \eta \) is on the higher side. The similarity score is calculated as:

\[
\text{Score}(G_1, G_2) = \alpha \cdot L(G_1, G_2) + \beta \cdot K(G_1, G_2)
\]

where \( L(G_1, G_2) \) is the Jaccard similarity defined as:

\[
L(G_1, G_2) = \frac{|E(G_1) \cap E(G_2)|}{|E(G_1) \cup E(G_2)|}
\]

and \( K(G_1, G_2) \) is the Kulczynski similarity:

\[
K(G_1, G_2) = \frac{1}{2} \left( \frac{|E(G_1) \cap E(G_2)|}{|E(G_1)|} + \frac{|E(G_1) \cap E(G_2)|}{|E(G_2)|} \right)
\]

In our case, \( \alpha = 0.5 \) and \( \beta = 0.5 \), giving equal weight to both similarities. For example, in our toy temporal network between two snapshots at \( t=1 \) and \( t=2 \) represented as \( G_1^1 \) and \( G_2^2 \) respectively, the number of common edges \( |E(G_1^1) \cap E(G_2^2)| = 1 \), \( |E(G_1^1)| = 1 \), and \( |E(G_2^2)| = 4 \), the Jaccard similarity would be:

\[
L(G_1, G_2) = \frac{4}{1 + 4 - 1} = \frac{1}{4} = 0.25
\]

and the Kulczynski similarity would be:

\[
K(G_1, G_2) = \frac{1}{2} \left( \frac{1}{1} + \frac{1}{4} \right) = \frac{1}{2} \left( 1 + 0.25 \right) = 0.6
\]

Thus, the overall similarity score is:

\[
\text{Score}(G_1, G_2) = 0.5 \cdot 0.25 + 0.5 \cdot 0.6 = 0.4
\]

Since the score exceeds \( \eta = 0.3 \), the snapshot 2 is selected. Similarly snapshots 4, and 8 are selected and hence we significantly reduce computation time by focusing on only few snapshots instead of all.

Next, we calculate the influence of each node. Initially, for each node \( i \), the influence probability \( p(i, t_0) \) is set to 1 for seed nodes and 0 otherwise. For each time interval \( [t_j, t_{j+1}] \), we compute the edge set \( E(t_j, t_{j+1}) \) and update the influence probabilities using the cpSI-R dynamic equations. For instance, if node 2 is a seed, its influence at time \( t_j \) is updated based on the spread matrix \( P(t_j, t_{j+1}) \), reflecting the probabilities of infecting its neighbors. After calculating individual node influences, we apply the Lazy Forward Influence algorithm to select the top \( k \) seed nodes, which maximizes the marginal gain in influence. In our toy network, we set \( k = 2 \), and after running the algorithm, nodes 2 and 3 emerge as the top seeds based on their cumulative influence.

The Lazy Forward Influence algorithm starts by calculating the influence of all nodes at the initial time point \( t_0 \) and constructs a priority queue of nodes based on their influence values. As the algorithm progresses, we iteratively update the seed set \( S \) to include nodes that offer the highest marginal gain. The influence of the final seed set \( S \) is computed as:

\[
\text{Final Influence}(S) = |S| + \sum_{i \in N(t) \setminus S(t)} p(i, t_r)
\]

where \( p(i, t_r) \) represents the influence probability of non-seed nodes at the final time stamp \( t_r \).

This approach not only reduces computational overhead but ensures accurate selection of influential nodes while capturing meaningful structural changes in the temporal network, resulting in nodes 2 and 3 being identified as the most influential. This is further demonstrated in the results section, where the optimal seed sets and influence spread are compared across different values of \( \eta \) and other parameters.

\begin{figure}[ht]
    \centering
    \subfigure[]{\includegraphics[width=0.46\textwidth]{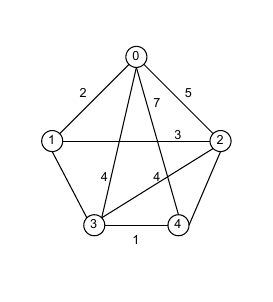}} 
    \subfigure[]{\includegraphics[width=0.52\textwidth]{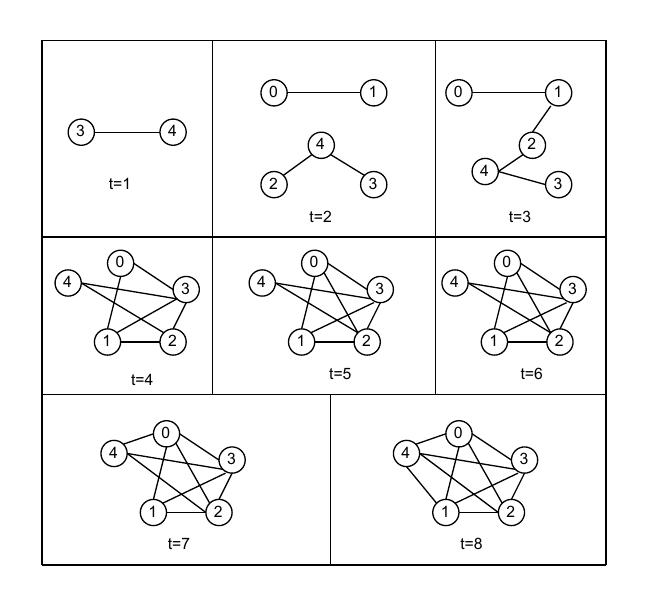}} 
    \caption{Example illustrating the working of cpSI-R model (a): 
    Toy temporal network represented with initial contact sequences  (b): Snapshots showing temporal evolution of a toy temporal network}
    \label{f4}
\end{figure}

\subsubsection{Computational Complexity}
The computational complexity of the \texttt{cpSI-R} model fundamentally depends on two components: the number of time intervals $T$ over the temporal network and the effort required to compute influence over these intervals. In a naive approach, one would simulate influence propagation across all $T$ snapshots. For each snapshot, the propagation of influence involves iterating over all temporal edges and updating influence probabilities, leading to a complexity of $\mathcal{O}(T \cdot |E|)$ for a complete forward simulation. Furthermore, for influence maximization using a greedy strategy over a seed set of size $k$, the marginal gain computation must be repeated for multiple candidate nodes, resulting in an added factor of $\mathcal{O}(k \log |V|)$ due to lazy evaluation. Thus, the overall worst-case time complexity in the absence of any optimization is:

\[
\mathcal{O}(T \cdot |E| + k \log |V| \cdot T)
\]

However, this full-scan approach is computationally prohibitive, especially for long temporal sequences. To mitigate this, the \texttt{cpSI-R} model incorporates a structural similarity-driven sampling algorithm that adaptively segments the temporal network into $r \ll T$ non-uniform intervals. These intervals are selected by identifying structural change points between consecutive snapshots using a combined similarity score derived from the Jaccard index:
\[
J(G_t, G_{t+\delta}) = \frac{|E_t \cap E_{t+\delta}|}{|E_t \cup E_{t+\delta}|}
\]
and the Kulczynski measure:
\[
K(G_t, G_{t+\delta}) = \frac{1}{2} \left( \frac{|E_t \cap E_{t+\delta}|}{|E_t|} + \frac{|E_t \cap E_{t+\delta}|}{|E_{t+\delta}|} \right)
\]
Snapshots are retained only when the similarity score falls below a pre-defined threshold $\theta$, indicating significant structural deviation. This selective sampling dramatically reduces the number of intervals to be processed, yielding a reduced set of $r$ time intervals that preserve key structural dynamics.

The influence estimation subroutine $\texttt{CalcInfluence}(S, r)$ then operates over only $r$ sampled intervals. Each forward simulation involves updating influence probabilities $\pi_v^{[t_i, t_{i+1}]}$ for each node $v$ using cpSI-R's dynamic equations within the interval $[t_i, t_{i+1}]$. For all temporal edges $|E|$ within these intervals, the influence propagation step runs in $\mathcal{O}(|E|)$ time per interval. Meanwhile, the influence maximization process applies a lazy greedy algorithm, exploiting submodularity to perform at most $\mathcal{O}(k \log |V|)$ marginal gain evaluations.

Combining these optimizations, the total time complexity of the model is significantly reduced to:
\[
\mathcal{O}(r \cdot |E| + k \log |V| \cdot r)
\]
This formulation achieves a favorable balance between computational efficiency and influence spread quality. By reducing the number of required influence propagation steps through intelligent sampling, the model scales effectively to large temporal networks while preserving theoretical guarantees of monotonicity and submodularity.

\section{Experiments}
\label{experiments}
To evaluate the effectiveness of the proposed approach, we conducted experiments on two real-world temporal network datasets. All experiments were executed on a standard computing environment with a multi-core processor and 16 GB of RAM. The implementation was carried out in Python 3.11 using the PyCharm development environment.

\subsection{Datasets}
We used the following data sets in order to test the approach.

\textbf{Primary school temporal network data \cite{gemmetto2014mitigation}\cite{stehle2011high}:} The temporal interactions between students and teachers are captured in this dataset. Many lines with the same timestamp indicate multiple encounters in the same span. Seconds are used to measure time. There are 125,775 edges and 242 nodes in the dataset.

\textbf{Contact patterns in a village in rural Malawi \cite{ruralmalawids}:} This dataset contains the contact patterns in a village in rural Malawi, based on proximity sensors technology. 

\textbf{CollegeMsg temporal network \cite{panzarasa2009patterns}:}This dataset is comprised of private messages sent on an online social network at the University of California, Irvine. There are 1899 nodes and 59835 temporal edges in the dataset.

\textbf{wiki-talk temporal network \cite{paranjape2017motifs}\cite{leskovec2010governance}:}This is a temporal network representing Wikipedia users editing each other's Talk page. A directed edge (u, v, t) means that user u edited user v's talk page at time t. The dataset consists of 1140149 nodes and 7833140 temporal edges.

\textbf{Stack Overflow temporal network \cite{paranjape2017motifs}:}This is a temporal network of interactions on the stack exchange web site Stack Overflow. The dataset consists of 2601977 nodes and 63497050 temporal edges

\textbf{Synthetic temporal network:} To simulate a scalable and evolving temporal interaction network, we generated a synthetic temporal network in Python using NetworkX and NumPy. The network comprises 3,000,000 nodes and over 70,000,000 temporal edges. The edge generation follows a modified preferential attachment mechanism, where newly added nodes are more likely to connect with high-degree existing nodes, emulating popularity-driven behaviour.

\begin{table}[htbp]
\centering
\caption{Summary of Temporal Network Datasets}
\begin{tabular}{|p{4cm}|p{2.5cm}|p{2.5cm}|p{5cm}|}
\hline
\textbf{Dataset} & \textbf{Number of Nodes} & \textbf{Number of Temporal Edges} & \textbf{Description} \\
\hline
Primary School Temporal Network \cite{gemmetto2014mitigation, stehle2011high} & 242 & 125,775 & Captures face-to-face interactions between students and teachers using temporal proximity data. \\
\hline
Rural Malawi Contact Network \cite{ruralmalawids} & 142 & 9,436 & Proximity-based contact patterns in a village in rural Malawi. Node and edge counts not specified. \\
\hline
CollegeMsg Network \cite{panzarasa2009patterns} & 1,899 & 59,835 & Temporal data of private messages exchanged on a university online social network. \\
\hline
Wiki-Talk Network \cite{paranjape2017motifs, leskovec2010governance} & 1,140,149 & 7,833,140 & Wikipedia users editing each other’s Talk pages, represented as directed temporal interactions. \\
\hline
Stack Overflow Network \cite{paranjape2017motifs} & 2,601,977 & 63,497,050 & Temporal network of user interactions on the Stack Overflow platform. \\
\hline
Synthetic Temporal Network & 3,000,000 & 70,000,000 &Temporal Barabási–Albert model \\
\hline

\end{tabular}
\label{tab:datasets}
\end{table}

\subsection{Implementation and Baselines}
We calculate the seed set using \textit{CalcInfluence()}, then utilize \textit{SamplingAlgorithm()} to determine the sampled time stamps, and finally apply \textit{TemporalInfluenceMaximization()} to calculate the influence of the seed set. Our approach relies on the sampling of temporal data to enhance accuracy in capturing influence maximization. We measured the influence  across different sampling rates and determined the optimal values in both datasets. This choice balances the trade-off between computational feasibility and the accuracy of influence estimation. We chose the sampling parameter \(\eta\) values of 0.7 and 0.6 across all datasets because empirical validation showed these values effectively capture significant structural changes in temporal networks without oversampling or undersampling. Additionally, we select \(\alpha = 0.5\) and \(\beta = 0.5\) to give equal importance to both local and global similarity.

We compared the performance of our approach with the following baselines:
\begin{itemize}
    \item \textbf{Dynamic Degree Discount \cite{murata2018extended}:}  The algorithm selects influential seed nodes in a dynamic network over a time period using a susceptibility parameter $\alpha$. The seed set is initially empty. Each node $j$ is assigned an initial degree discount value $\Delta_j = \mathcal{D}_\mathcal{T}(j)$ and a time counter $\tau_j = 0$. During each iteration, the node with the highest $\Delta_j$ is added to the seed set. For each neighbor $q$ of this node, the degree discount and time counter are updated as $\tau_q = \tau_q + 1$ and $\Delta_q = \mathcal{D}_\mathcal{T}(q) - 2 \tau_q - (\mathcal{D}_\mathcal{T}(q) - \tau_q) \tau_q \alpha$. The algorithm continues until the seed set reaches the desired size. We have chosen the value of $\alpha = 0.01$ in our experiments.
\item \textbf{Borgs and Tangs Algorithm \cite{borgs2014maximizing} \cite{tang2014influence}:}  This method by Borgs and Tang aims to identify influential seed nodes in a fixed network \( \mathcal{N} \). The algorithm starts by selecting a node \( w \) at random and removing edges from \( \mathcal{N} \) with probability \( 1 - \lambda \), where \( \lambda \) is the susceptibility parameter. Afterward, it gathers nodes that are still reachable from \( w \), forming a set of reachable nodes \( \mathcal{R} \). This process is repeated \( \gamma \) times. The node with the highest frequency in \( \mathcal{R} \) is added to the seed set \( \mathcal{S} \). The procedure continues until the size of \( \mathcal{S} \) reaches \( k \).
In our experiments, we have chosen the value of susceptibility as $0.01$
\item \textbf{Dynamic CI \cite{murata2018extended}:}  Dynamic CI expands Morone's method to dynamic networks by redefining node degrees and distances over time. The Dynamic CI index, \( D_{\Theta_m}(x) \), is calculated as:

\[
D_{\Theta_m}(x) = \Gamma(x) \sum_{y \in \Delta(x,m)} \Gamma(y)
\]

where \( \Gamma(x) \) is the dynamic degree of node \(x\), and \( \Delta(x,m) \) represents the set of nodes reachable from \(x\) within a minimum time of \(m\). This index helps identify the top \(k\) influential nodes.Since it is difficult to perform experiments on taking different values of $m$, we have used $m=10$ in our experiments.

\item \textbf{Forward Influence Algorithm \cite{charuagrawal}:} The forward influence algorithm identifies influential nodes by first computing $\text{DynInfuenceVal}(\{i\}, F(\cdot), (t_0, t_0 + h))$ for each node $i \in N(t_0)$ and selecting the top $k$ nodes. It then iteratively replaces nodes in this set with others to maximize the increase in influence, evaluating changes as $\text{DynInfuenceVal}(S \cup \{i\} \setminus \{j\}, F(\cdot), (t_0, t_0 + h)) - \text{DynInfuenceVal}(S, F(\cdot), (t_0, t_0 + h))$. The process continues until no significant improvement is found or a maximum number of iterations is reached.

\item \textbf{INMFA \cite{Erkol_2020}:}  INMFA models the temporal SIR dynamics using marginal probabilities $S_i^{(t)}, I_i^{(t)}, R_i^{(t)}$, with the constraint $S_i^{(t)} + I_i^{(t)} + R_i^{(t)} = 1$. The evolution is given by:
    \[
    I_i^{(t)} = (1 - \mu) I_i^{(t-1)} + \left(1 - I_i^{(t-1)} - R_i^{(t-1)}\right) \left[1 - \prod_j \left(1 - \lambda A_{ji}^{(t-1)} I_j^{(t-1)}\right)\right],
    \]
    \[
    R_i^{(t)} = R_i^{(t-1)} + \mu I_i^{(t-1)},
    \]
    where $\lambda$ is the infection rate, $\mu$ is the recovery rate, and $A_{ji}^{(t)}$ is the temporal adjacency matrix. INMFA neglects inter-node correlations, leading to a systematic overestimation of spread size $O^{(t)}_{\mathrm{INMFA}} = \frac{1}{N} \sum_i \left(I_i^{(t)} + R_i^{(t)}\right)$. Despite this, it offers good approximation near the critical regime and is computationally efficient.

\item \textbf{Entropy Based \cite{michalski2020entropy}:}This method quantifies the temporal diversity of a node’s neighborhood using entropy. For a node $v_i$, the variability is defined as:
    \[
    \mathrm{Var}(v_i) = -\sum_{n=1}^{N-1} p(\mathrm{NB}_{v_i,n,n+1}) \log p(\mathrm{NB}_{v_i,n,n+1}),
    \]
    where $p(\mathrm{NB}_{v_i,n,n+1})$ is the probability of a specific neighbor configuration across consecutive time windows, and $|e_{ij}|$ denotes the total unique edges from $v_i$. This node-level metric captures dynamic neighborhood changes efficiently, making it scalable for large temporal networks. Top-ranked nodes by entropy are selected as seeds, typically during the central window of the timeline.

\item \textbf{TBCELF \cite{tbcelf}:}TBCELF is a budget-aware influence maximization approach for temporal networks that combines cost-effective lazy forward optimization with temporal seed selection. It ensures high-quality seed selection within budget constraints

\end{itemize}

\subsection{Experimental Results}

We conduct a comprehensive evaluation of our proposed temporal influence maximization method using all the datasets discussed above. The key metrics analyzed are infection spread and execution time across varying seed set sizes. All infection spread and seed set values are reported as percentages for uniform comparison.

We begin by evaluating our method on the Primary School dataset, divided into seven temporal snapshots, each with approximately 1000 timestamps. Figure~\ref{snapshot-viz} visualizes five of these snapshots. We tested our method under two diffusion models: SIR and cpSI-R, varying the seed set size as $k = 10, 20, 30$. From the results illustrated in Figure~\ref{fig:model-compare}, we observe that during the early evolution of the temporal network, both models perform similarly. However, in later snapshots—particularly snapshot 7—the cpSI-R model significantly outperforms SIR in terms of spread. This can be attributed to cpSI-R’s reactivation mechanism, which allows previously recovered nodes to become active again with a fading probability. Such reactivated nodes, often situated in densely connected regions, initiate stronger cascades, resulting in higher infection spreads. Notably, the performance gap between cpSI-R and SIR widens as the seed set size increases.
\begin{figure}[htbp]
    \centering
    \includegraphics[width=0.99\textwidth]{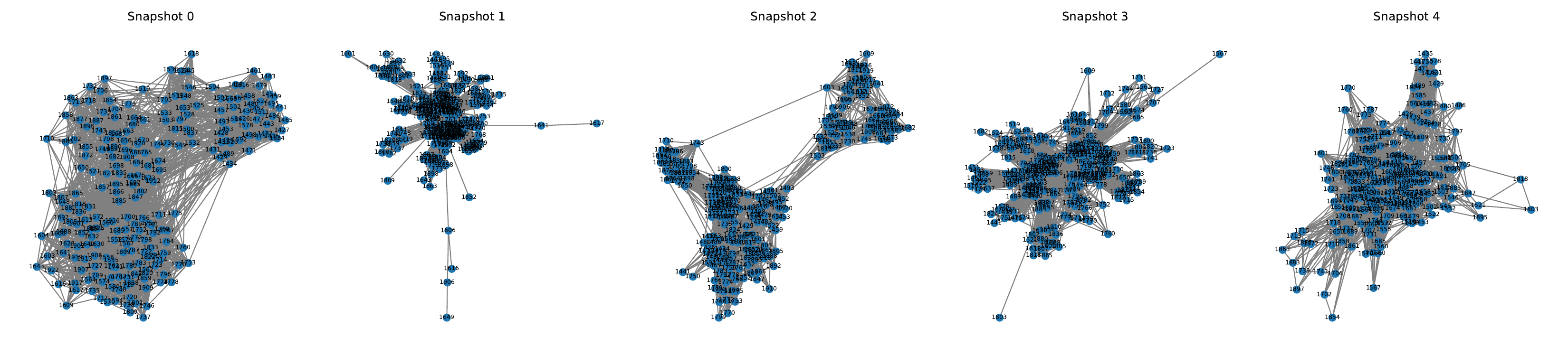}
    \caption{ Five snapshots of a Primary School temporal network sampled at 1.1K timestamps in each snapshot.}
    \label{snapshot-viz}
\end{figure}

\begin{figure}[htbp]
    \centering
    \includegraphics[width=0.99\textwidth]{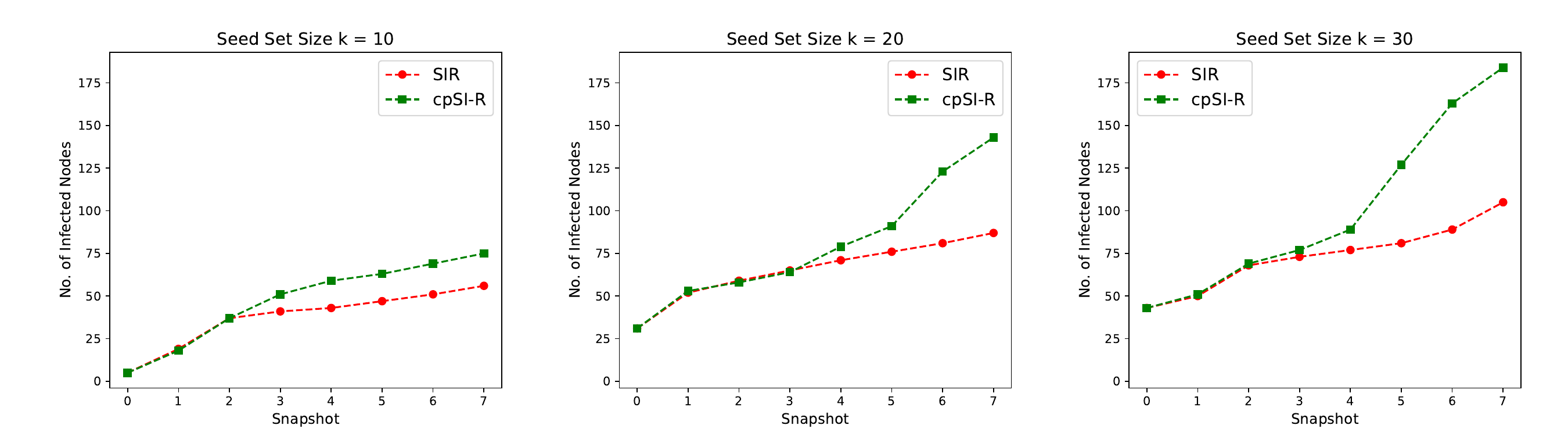}
    \caption{Advantage of cpSI-R model over SIR in terms of infection spread on Primary school dataset across various snapshots and seed set size}
    \label{fig:model-compare}
\end{figure}

To further illustrate the impact of reactivation, we designed a small-scale experiment on a toy temporal network (Figure~\ref{fig:toy-network}), comprising five evolving snapshots from $t = 0$ to $t = 4$. This network simulates a growing social system, where nodes and edges increase over time. Our results on this setup show that cpSI-R achieves an infection spread of 83.3\%, compared to 69.07\% with SIR. The reason is that, under cpSI-R, recovered nodes can re-engage in influence diffusion if they receive sufficient reinforcement, unlike SIR where such nodes become permanently inactive.
\begin{figure}[htbp]
    \centering
    \includegraphics[width=0.95\textwidth]{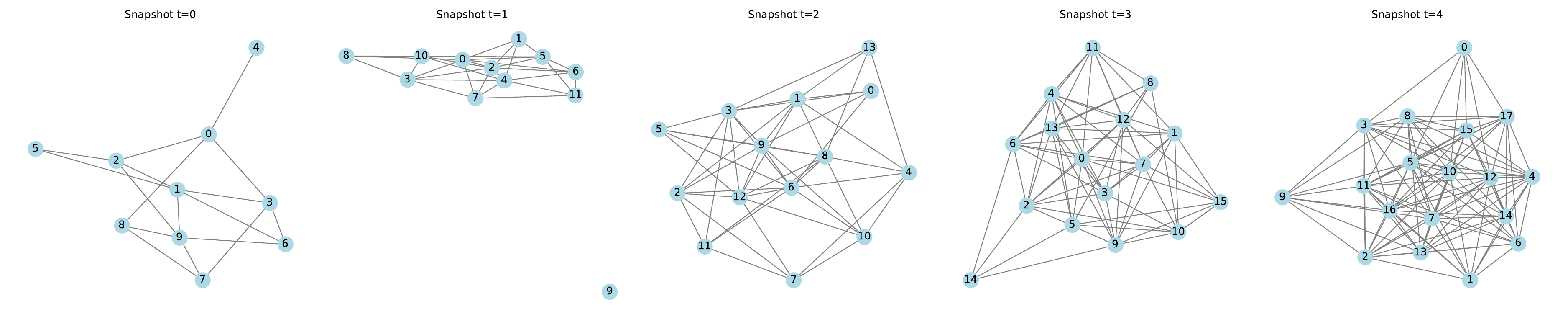}
    \caption{Toy temporal network illustrating the advantage of cpSI-R model over SIR in densely connected environment.}
    \label{fig:toy-network}
\end{figure}

In Figure~\ref{fig:eta-varying}, we examine the effect of the sampling parameter $\eta$ on infection spread across all datasets. We find that spread is maximized at $\eta = 0.7$ and $0.8$, and drops sharply for lower values. This implies that as we incorporate more temporal dynamics (lower $\eta$), the influence spread diminishes. Conversely, as we move towards a static setting (higher $\eta$), the number of infections increases. To determine the optimal $\eta$, we also evaluated execution time (Figure~\ref{fig:eta-time}). Results show a steep increase in time cost beyond $\eta = 0.7$, confirming that denser snapshots require more computation. Thus, we fix $\eta = 0.7$ for the rest of our experiments, balancing spread and runtime efficiency.
\begin{figure}[H]  
    \centering
    \includegraphics[width=0.95\textwidth]{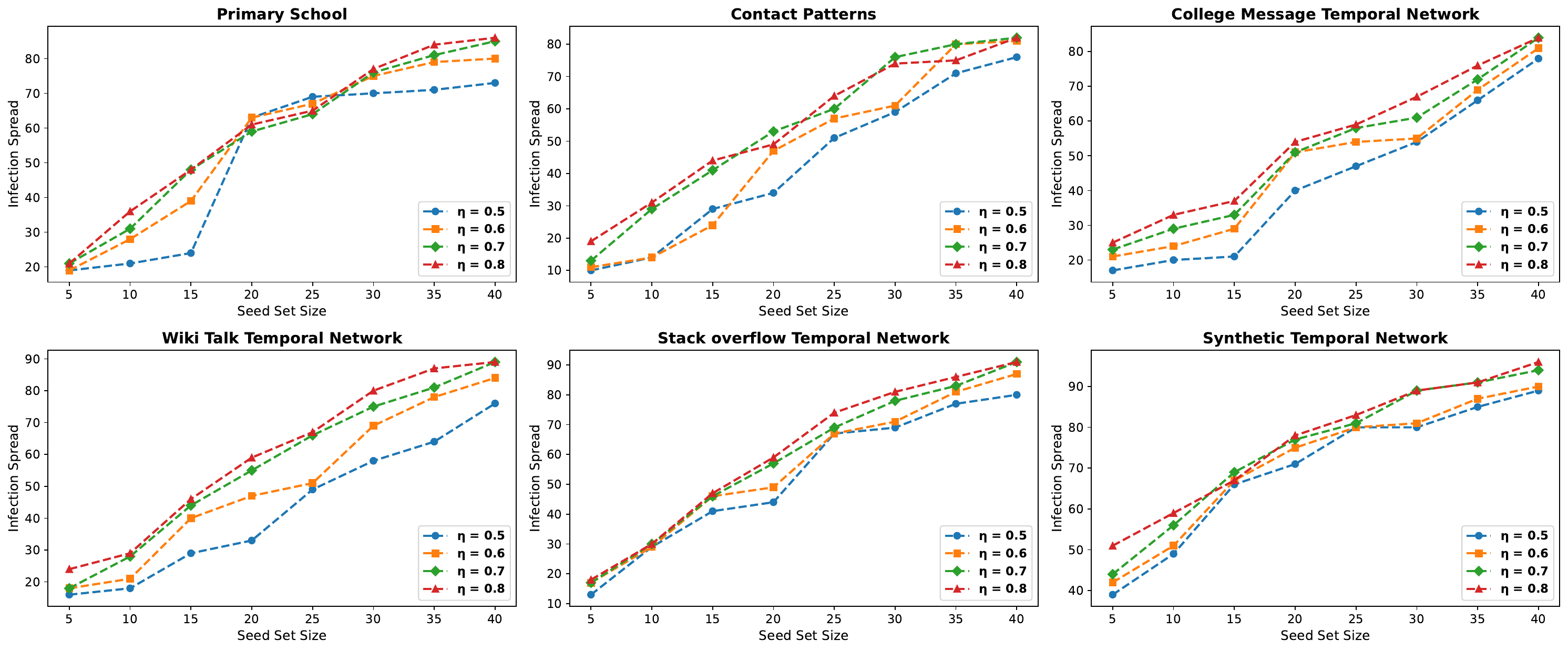}
    \caption{Comparative infection spread across multiple datasets for different values of $\eta$ using our method and cpSI-R model.}
    \label{fig:eta-varying}
\end{figure}

\vspace{-5mm}  

\begin{figure}[H]
    \centering
    \includegraphics[width=0.95\textwidth]{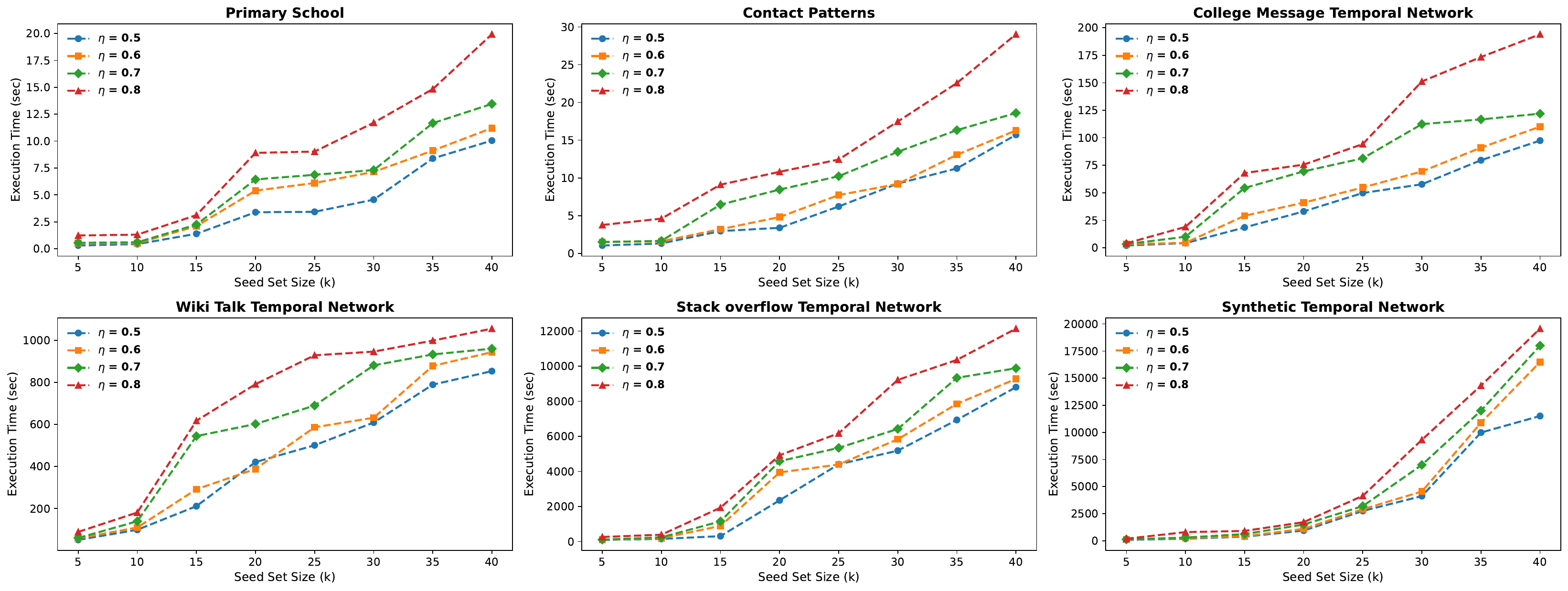}
    \caption{Comparative execution time (sec) across multiple datasets for our influence maximization method and cpSI-R model on changing $\eta$ values.}
    \label{fig:eta-time}
\end{figure}

\FloatBarrier

Using this optimal configuration, we compared our method to all baseline approaches. All methods use the cpSI-R model for fair comparison. As shown in Figure~\ref{fig:spread-compare}, our method consistently achieves higher infection spread, particularly for larger seed set sizes and in large-scale datasets like WikiTalk and StackOverflow. In Figure~\ref{fig:time-compare}, we observe that due to our efficient temporal sampling and optimization, our method also yields lower execution time as seed set size increases, outperforming all baselines across all datasets.
\begin{figure}[H]
    \centering
    \includegraphics[width=0.95\textwidth]{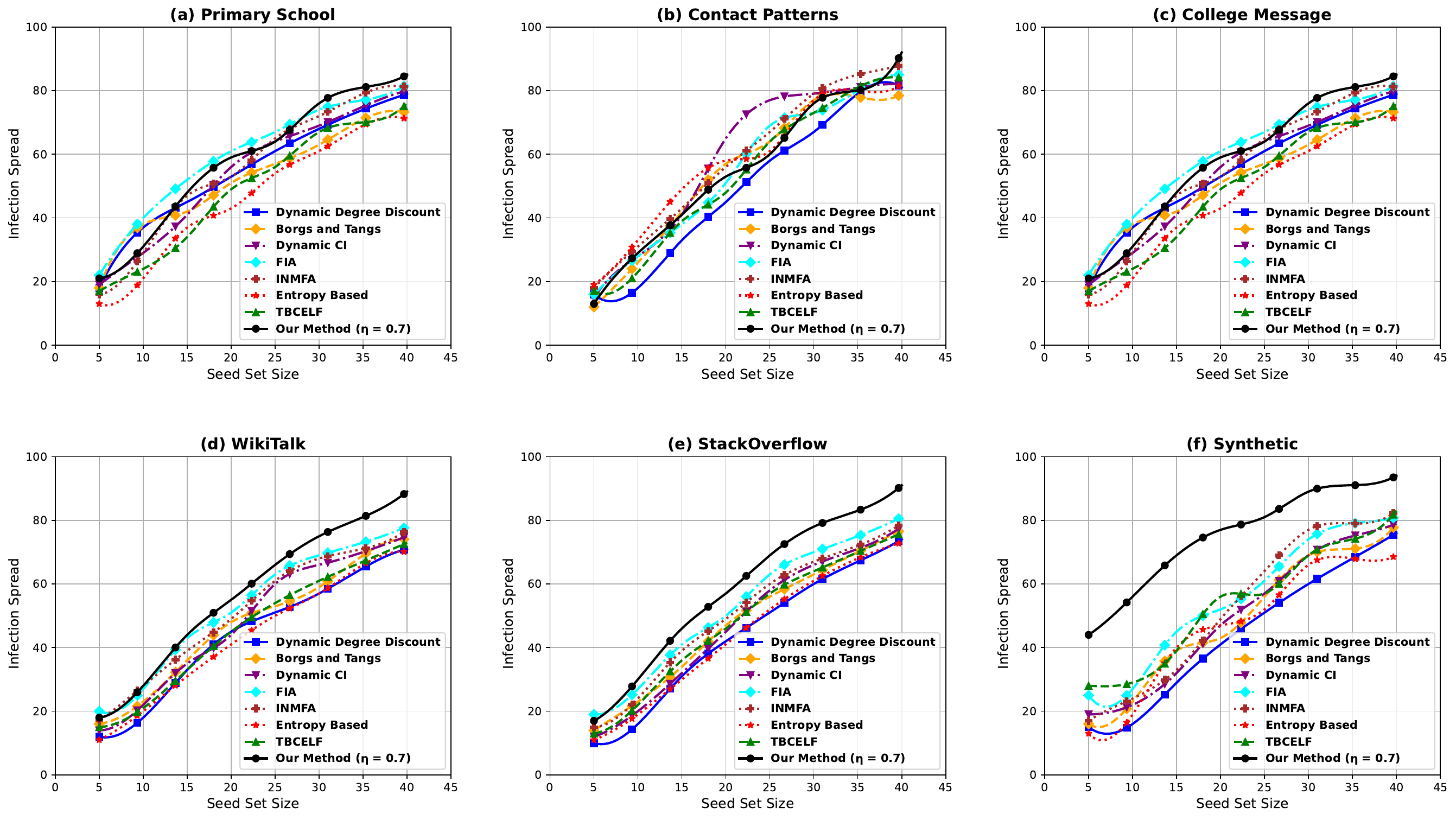}
    \caption{Comparative infection spread across multiple datasets for different influence maximization methods. Our method ($\eta=0.7$) consistently outperforms others across all settings.}
    \label{fig:spread-compare}
\end{figure}

\vspace{-5mm} 

\begin{figure}[H]
    \centering
    \includegraphics[width=0.95\textwidth]{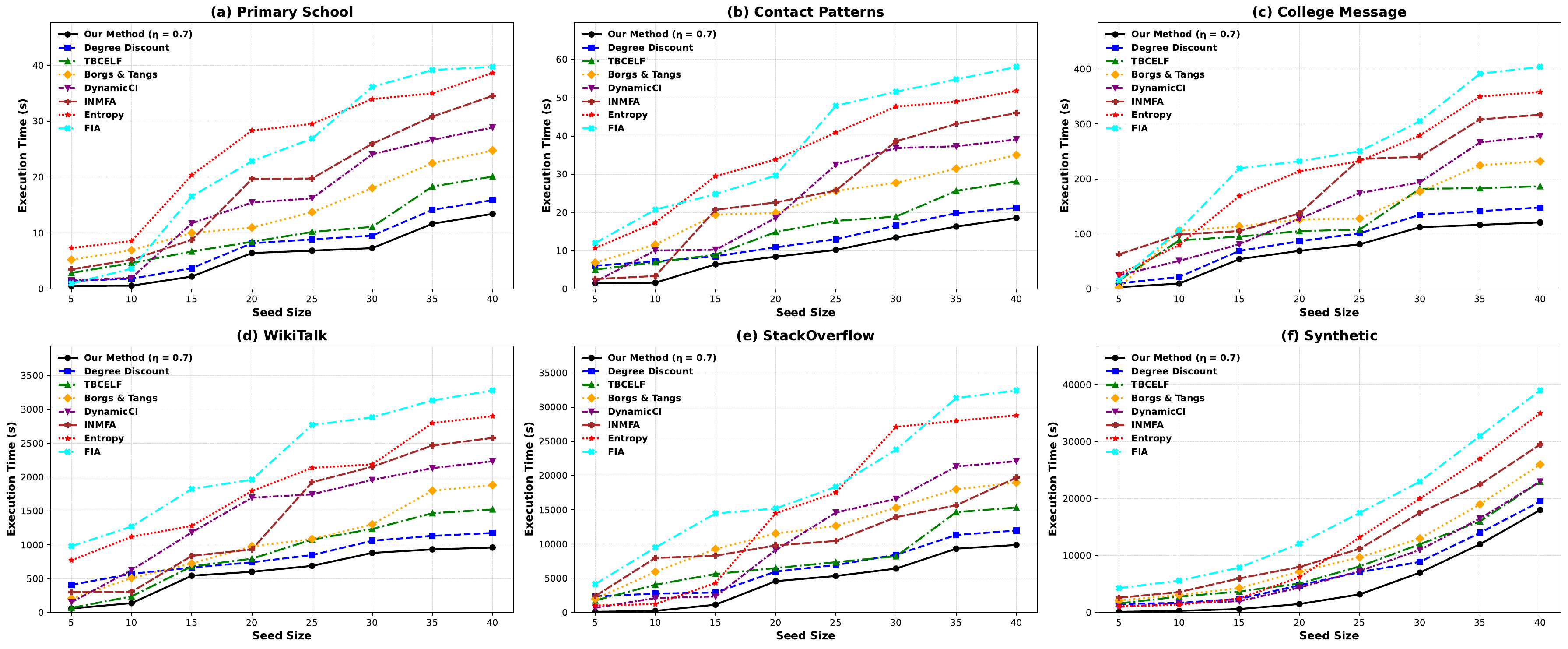}
    \caption{Execution time in seconds across multiple datasets for different influence maximization methods. Our method ($\eta=0.7$) consistently outperforms others across all settings.}
    \label{fig:time-compare}
\end{figure}

\section{Discussion}
\label{discuss}

The experimental evaluation validates the robustness and practicality of the proposed cpSI-R model and its corresponding influence maximization framework. The key advantage of cpSI-R lies in its incorporation of behavioral dynamics, specifically, temporary disengagement and reactivation of nodes, as well as reinforcement through persistent exposure. These features allow the diffusion process to better mirror real-world influence dynamics where individuals may repeatedly re-engage with information over time.

The comparative study between cpSI-R and the classical SIR model (Figure~\ref{fig:model-compare}) demonstrates a marked improvement in influence spread, especially as the diffusion progresses through time. This substantiates the claim that reactive and persistent behaviors enhance the ability of the model to sustain influence propagation in temporal settings.

A pivotal element of our framework is the structure-aware adaptive temporal sampling strategy governed by the parameter $\eta$. The results indicate that $\eta = 0.7$ strikes the optimal balance between capturing local bursty interactions and global temporal dependencies. At this value, the infection spread is maximized across datasets while maintaining manageable computational cost (Figures~\ref{fig:eta-varying} and~\ref{fig:eta-time}). This suggests that $\eta = 0.7$ offers a principled trade-off between accuracy and scalability.

When benchmarked against existing influence maximization methods (Figures~\ref{fig:spread-compare} and~\ref{fig:time-compare}), our method consistently outperforms in terms of infection spread across all seed sizes and network types. Importantly, this performance gain does not come at the cost of excessive runtime. While more computationally involved than simple heuristics, our approach is significantly more efficient than information-theoretic and probabilistic baselines, especially on large-scale networks.

We note, however, that on smaller temporal networks, our method may occasionally trail lightweight heuristics in runtime or marginal spread. This is primarily because smaller graphs typically exhibit low structural variability across snapshots, making simple greedy or degree-based heuristics highly effective without requiring sophisticated sampling or behavioral modeling. Additionally, the overhead introduced by temporal sampling and reinforcement tracking in cpSI-R may not yield substantial marginal gains when the number of nodes and edges is limited. Nevertheless, as the network size and temporal granularity increase, our framework scales efficiently and its advantages become more pronounced. These findings reinforce the central premise of this work: combining behavioral modeling (via cpSI-R), temporal awareness, and submodular optimization yields a scalable and effective influence maximization strategy that is well-suited to the complex dynamics of real-world temporal networks.

\section{Conclusion and Future Work}
\label{conclusion}

The cpSI-R model introduces a substantial advancement in influence maximization over temporal social networks by incorporating key behavioral dynamics such as persistence and time-limited reactivation. These mechanisms allow cpSI-R to more accurately model real-world diffusion processes where individuals may temporarily disengage and later rejoin due to renewed influence. The theoretical guarantees of monotonicity and submodularity make the model well-suited for optimization, enabling the design of more effective seed selection strategies across varying network conditions. To operationalize cpSI-R, we developed an efficient influence maximization framework, including a tailored seed selection mechanism and a novel temporal sampling algorithm. Our comprehensive experiments on  real-world temporal datasets and synthetic dataset demonstrate that the proposed approach not only achieves consistently higher influence spread but also offers superior computational efficiency. These improvements are particularly evident on large-scale networks, where our method significantly outperforms competitive baselines in both spread quality and running time.

Looking ahead, we aim to enhance the adaptability of cpSI-R to a wider range of temporal network structures, including those characterized by heterogeneous interaction frequencies or irregular temporal patterns. Future research will also explore adaptive and real-time seed selection strategies that can respond dynamically to evolving network states. Additionally, extending cpSI-R to practical domains such as epidemic control, political mobilization, and targeted marketing will further validate its robustness and applicability. Interdisciplinary collaboration will be instrumental in broadening the theoretical scope and real-world impact of the cpSI-R framework.

\newpage
\bibliographystyle{ieeetr} 
\bibliography{Refr}
\end{document}